%% file: main.tex
\documentclass[USenglish,letterpaper,11pt]{article}

\usepackage[T1]{fontenc}
\usepackage[margin=1in]{geometry}
\usepackage{float,xspace}
\usepackage{amsmath, amsthm, amssymb}
\usepackage[nottoc, notlof, notlot]{tocbibind}
\usepackage{graphicx}
\usepackage{mathtools}
\usepackage{enumitem}

\usepackage[pagebackref=true]{hyperref}
\hypersetup{
    pdftitle={Quantum and Classical Algorithms for Approximate Submodular Function Minimization},
    pdfauthor={Yassine Hamoudi, Patrick Rebentrost, Ansis Rosmanis, Miklos Santha},
    colorlinks=true, linkcolor=red, citecolor=blue, filecolor=blue, urlcolor=blue
}

\usepackage{algorithm}
\usepackage[noend]{algpseudocode}
\newcommand{\algparbox}[1]{\parbox[t]{\dimexpr\linewidth-\algorithmicindent}{#1\strut}} 
\algnewcommand{\LineComment}[1]{\Statex \(\triangleright\) #1} 

\usepackage{tikz}
\usetikzlibrary{tikzmark,shapes.multipart}
\newcommand*\circled[1]{\tikz[baseline=(char.base)]{\node[shape=circle,draw,inner sep=1.6pt] (char) {#1};}}

\theoremstyle{plain}
\newtheorem{theorem}{Theorem}[section]
\newtheorem{proposition}[theorem]{Proposition}
\newtheorem{corollary}[theorem]{Corollary}
\newtheorem{lemma}[theorem]{Lemma}
\newtheorem{fact}[theorem]{Fact}

\theoremstyle{definition}
\newtheorem{definition}[theorem]{Definition}

\newtheorem{assumption}{Assumption}

\theoremstyle{remark}

\newcommand{\R}{\mathbb{R}}

\newcommand{\ra}{\rightarrow}
\newcommand{\rn}{\{0,1\}}
\newcommand{\eps}{\epsilon}
\newcommand\super[1]{^{(#1)}}
\DeclareMathOperator*{\argmin}{\arg\!\min}
\DeclareMathOperator*{\argmax}{\arg\!\max}
\DeclareMathOperator{\sgn}{sgn}

\DeclarePairedDelimiterX{\inp}[2]{\langle}{\rangle}{#1, #2}
\DeclarePairedDelimiterX{\norm}[1]{\lVert}{\rVert}{#1}
\DeclarePairedDelimiterX{\abs}[1]{|}{|}{#1}
\DeclarePairedDelimiterX{\parentheses}[1]{\lparen}{\rparen}{#1}

\newcommand{\bo}{O\parentheses}
\newcommand{\so}{\widetilde{O}\parentheses}
\newcommand{\zo}{O^{*}\parentheses}

\newcommand{\spa}[1]{\textup{Span}\parentheses}
\DeclarePairedDelimiterX{\bra}[1]{\langle}{|}{#1}
\DeclarePairedDelimiterX{\ket}[1]{|}{\rangle}{#1}
\DeclarePairedDelimiterX{\bracket}[3]{\langle}{\rangle}{#1|#2|#3}
\DeclarePairedDelimiterX{\ip}[2]{\langle}{\rangle}{#1|#2}

\DeclarePairedDelimiterX{\expectarg}[1]{[}{]}{\activatebar#1} 
\newcommand{\innermid}{\nonscript\;\delimsize\vert\allowbreak\nonscript\;}
\newcommand{\activatebar}{\begingroup\lccode`\~=`\| \lowercase{\endgroup\let~}\innermid \mathcode`|=\string"8000}
\newcommand{\pb}{\Pr\expectarg} 
\newcommand{\ex}{\mathbb{E}\expectarg} 

\newcommand{\fl}{f} 
\newcommand{\gl}{g} 
\newcommand{\eo}{\mathrm{EO}}

\newcommand{\lv}{Lov\'asz}
\newcommand{\dis}[1]{\mathcal{D}_{#1}}
\newcommand{\xopt}{x^{\star}}
\newcommand{\wg}{\widetilde{g}}
\newcommand{\hg}{\widehat{g}}
\newcommand{\wwg}{\widetilde{\raisebox{0pt}[0.85\height]{$\widetilde{g}$}}{}}
\newcommand{\wdi}{\widetilde{d}}
\newcommand{\gd}{\mathbf{V}}
\newcommand{\ps}{_+}
\newcommand{\ms}{_-}
\newcommand{\samplg}{\textup{\textsf{GSample}}}
\newcommand{\sampld}{\textup{\textsf{GDSample}}}
\newcommand{\cgs}{\mathcal{C}\text{-}\textup{\textsf{GS}}}
\newcommand{\cgds}{\mathcal{C}\text{-}\textup{\textsf{GDS}}}
\newcommand{\qgs}{\mathcal{Q}\text{-}\textup{\textsf{GS}}}
\newcommand{\qgds}{\mathcal{Q}\text{-}\textup{\textsf{GDS}}}
\newcommand{\costg}{c_{\textup{\tiny\textsf{GS}}}}
\newcommand{\costd}{c_{\textup{\tiny\textsf{GDS}}}}
\newcommand{\ds}{\textup{\textsf{D}}} 

\title{Quantum and Classical Algorithms for \\ Approximate Submodular Function Minimization}

\author{Yassine Hamoudi\thanks{Universit\'e de Paris, IRIF, CNRS, F-75013 Paris, France.}
   \and Patrick Rebentrost\thanks{Centre for Quantum Technologies, National University of Singapore, Singapore 117543.}
   \and Ansis Rosmanis\footnotemark[2]
   \and Miklos Santha\thanks{Universit\'e de Paris, IRIF, CNRS, F-75013 Paris, France;  and Centre for Quantum Technologies and  MajuLab UMI 3654, National University of Singapore, Singapore 117543.}}

\date{}

\begin{document}

\maketitle
\thispagestyle{empty}
\setcounter{page}{0}


\begin{abstract}
  Submodular functions are set functions mapping every subset of some ground set of size $n$ into the real numbers and satisfying the diminishing returns property. Submodular minimization is an important field in discrete optimization theory due to its relevance for various branches of mathematics, computer science and economics. The currently fastest strongly polynomial algorithm for exact minimization~\cite{LSW15c} runs in time $\so{n^3 \cdot \eo + n^4}$ where $\eo$ denotes the cost to evaluate the function on any set. For functions with range $[-1,1]$, the best $\eps$-additive approximation algorithm~\cite{CLSW17c} runs in time $\so{n^{5/3}/\eps^{2} \cdot \eo}$.

  In this paper we present a classical and a quantum algorithm for approximate submodular minimization. Our classical result improves on the algorithm of \cite{CLSW17c} and runs in time $\so{n^{3/2}/\eps^2 \cdot \eo}$. Our quantum algorithm is, up to our knowledge, the first attempt to use quantum computing for submodular optimization. The algorithm runs in time $\so{n^{5/4}/\eps^{5/2} \cdot \log(1/\eps) \cdot \eo}$. The main ingredient of the quantum result is a new method for sampling with high probability $T$ independent elements from any discrete probability distribution of support size $n$ in time $\bo{\sqrt{Tn}}$. Previous quantum algorithms for this problem were of complexity $\bo{T\sqrt{n}}$.
\end{abstract}

\newpage


\section{Introduction}

  \input{Introduction.tex}

  \subsection{Previous Work}
  \input{PreviousWork.tex}

  \subsection{Our Contributions}
  \input{Contributions.tex}

  \subsection{Organization of the Paper}

  Our algorithms are presented in a modular way. In Section \ref{Sec:frame}, we give the common framework to the classical and quantum algorithms. Then, we specialize it to each setting in Sections \ref{Sec:classApprox} and \ref{Sec:quantApprox} respectively. A data structure, which is common to both models, is described in Section \ref{Sec:datastruc}. The robustness of the stochastic subgradient descent (Proposition \ref{Prop:gradDescent}), and the quantum algorithm for sampling from discrete probability distributions (Section \ref{Sec:Sampling}) can be read independently from the rest of the paper. The reader interested only in the classical algorithm can skip Sections \ref{Sec:Sampling} and \ref{Sec:quantApprox}.

  \subsection{Recent Improvement}

  After having finished this paper, we have been informed through personal communication that Axelrod, Liu and Sidford have discovered a classical nearly linear time algorithm for approximate submodular function minimization \cite{ALS20c}. Their result, like ours, improves on the work of Chakrabarty et al.~\cite{CLSW17c}, and it outperforms both our algorithms.

  \subsection{Open Questions}

  It seems to us that, in order to achieve a quantum speed-up over the best classical algorithms for approximate \cite{ALS20c} or exact \cite{LSW15c} submodular function minimization, one would most likely have to speed-up the gradient descent or cutting plane methods respectively. This latter problem is notoriously open in the quantum setting. Another more amenable question is  whether the $\Omega(n)$ lower bound for exact minimization \cite{Har08d} carries over to the quantum oracle model, and whether $\Omega(n/\eps^2)$ is a lower bound in the approximate case. Finally, what can be other applications of our quantum multi-sampling algorithms beyond submodular function minimization?

\section{Preliminaries}
\label{Sec:Prelim}
\input{Preliminaries.tex}

\section{Quantum Multi-Sampling for Discrete Probability Distributions}
\label{Sec:Sampling}
\input{Sampling.tex}

\section{Framework for Approximate Submodular Minimization}
\label{Sec:frame}
\input{Framework.tex}

\section{Subquadratic Approximate Submodular Minimization}

We construct two classical (Section \ref{Sec:classApprox}) and two quantum (Section \ref{Sec:quantApprox}) procedures satisfying the Assumptions \ref{Assp:one} and \ref{Assp:two} described in the previous section. These procedures will require a particular data structure, strongly inspired from the work of \cite{CLSW17c}, that is maintained throughout the subgradient descent algorithm at negligible cost (Section \ref{Sec:datastruc}). Our final algorithms solve the approximate submodular minimization problem in time $\so{n^{3/2}/\eps^2 \cdot \eo}$ classically, and $\so{n^{5/4}/\eps^{5/2} \cdot \log(1/\eps) \cdot \eo}$ quantumly.

  \subsection{Data structures and $k$-Covers}
  \label{Sec:datastruc}
  \input{Datastructure.tex}

  \subsection{Classical Approximate Submodular Minimization}
  \label{Sec:classApprox}
  \input{ClassicalApproximation.tex}

  \subsection{Quantum Approximate Submodular Minimization}
  \label{Sec:quantApprox}
  \input{QuantumApproximation.tex}

\section*{Acknowledgements}

The work of P.R., A.R. and M.S. was supported by the Singapore National Research Foundation, the Prime Minister's Office, Singapore and the Ministry of Education, Singapore under the Research Centres of Excellence programme under research grant R 710-000-012-135. The work of Y.H. and M.S. was supported by the QuantERA ERA-NET Cofund project QuantAlgo and the ANR project ANR-18-CE47-0010 QUDATA. This work was partly done while Y.H. was visiting CQT.


\newpage
\bibliography{Bibliography}

\appendix

  \section{Counterexample to the Gradient Sampling in \cite{CLSW17c}}
  \label{App:counter}
  \input{App_Counterexample.tex}


\end{document}

%% file: Introduction.tex
\subsection{Submodular Minimization}

A submodular function $F$ is a function mapping every subset of some finite set $V$ of size $n$ into the real numbers and satisfying the diminishing returns property: for every $A \subseteq B \subseteq V$ and for every $i \not\in B$, the inequality $F(A \cup \{i\}) - F(A) \geq F(B \cup \{i\}) - F(B)$ holds. In words, given two sets where one of them contains the other, adding a new item to the smaller set increases the function value at least as much as adding that element to the bigger set. Many classical functions in mathematics, computer science and economics are submodular, the most prominent examples include entropy functions, cut capacity functions, matroid rank functions and utility functions. Applications of submodular functions, or slight variants of them, occur in areas as far reaching as machine learning~\cite{Bac13j, KC10c, NKA11c}, operations research~\cite{IB13c, QS95c}, electrical networks~\cite{Nar09b}, computer vision~\cite{Hoc01j}, pattern analysis~\cite{BVZ01j} and speech analysis~\cite{LB10p}.

Submodular functions show analogies both with concavity and convexity. The diminishing returns property makes them akin to concave functions, but they have algorithmic properties similar to convex functions. In particular, while it follows from the NP-hardness of maximum cut that submodular maximization is NP-hard, submodular minimization can be solved in polynomial time, in fact even in strongly polynomial time. The link between submodular functions and convex analysis is made explicit through the \lv\ extension~\cite{Lov82c}. There are various approaches to solve submodular minimization. The foundational work of  Gr\"otschel, \lv\ and Schrijver~\cite{GLS81j} gave the first polynomial time algorithm using the ellipsoid method. The first pseudo-polynomial algorithm using a combinatorial method appeared in the influential paper of Cunningham~\cite{Cun85j}. In a later work Gr\"otschel, \lv\ and Schrijver~\cite{GLS88b} were the first to design a strongly polynomial time algorithm, and the first strongly polynomial time combinatorial algorithms were given by Schrijver~\cite{Sch00j} and by Iwata, Fleischer and Fujishige~\cite{IFF01j}. Many of these works assume an access to an evaluation oracle for the function $F$, where the time of a query is denoted by $\eo$.

The current fastest submodular minimization algorithm is by Lee, Sidford and Wong~\cite{LSW15c}. Their weakly polynomial algorithm runs in time $\so{n^2 \log M \cdot \eo + n^3 \log^{O(1)}M}$ and their strongly polynomial algorithm runs in time $\so{n^3 \cdot \eo + n^4}$, where $M$ is an upper bound on the integer valued function and the notation $\so{}$ hides polylogarithmic factors in $n$. Both algorithms apply a new cutting plane method given in the same paper and use the \lv\ extension. Our work is most closely related to the recent paper of Chakrabarty et al.~\cite{CLSW17c} who gave an $\so{nM^3 \cdot \eo}$ algorithm in the setting of~\cite{LSW15c}, and an $\eps$-additive approximation algorithm that runs in time $\so{n^{5/3}/\eps^2 \cdot \eo}$ for real valued submodular functions with range $[-1,1]$. These algorithms were the first to run in subquadratic time in $n$, by going beyond the direct use of the subgradients of the \lv\ extension. Indeed, it is proven in~\cite{CLSW17c} that any algorithm accessing only subgradients of the \lv\ extension has to make $\Omega(n^2)$ queries. Such algorithms include~\cite{LSW15c} and the Fujishige-Wolfe algorithm~\cite{Fuj80j, Wol76j}. Subquadratic approximate algorithms (such as in \cite{CLSW17c} or our work) can potentially lead to insights to the exact case. They are also more practical when the scaling in $n$ is more important than whether or not the algorithm is exact.

\subsection{Quantum Algorithms for Optimization}

A quite successful recent trend in quantum computing is to design fast quantum algorithms for various optimization and machine learning problems. At a high level, these algorithms are constructed in various subtly different input/output models \cite{Aar15j,BWP+17j,CHI+18j,vAG19c}. In the model that we are working with in this paper, the input is given by an oracle that can be accessed in quantum superposition, and the output is classical. The algorithms in this model are often hybrid, that is partly of classical and partly of quantum nature, and designed in a modular way so that the quantum part of the algorithm can be treated as a separate building block. In fact, a standard feature of these algorithms is that they make a quantum improvement on some part of the (best) available classical algorithm, but they keep its overall structure intact. In most cases the quantum versus classical speed-up is at most polynomial, usually at most quadratic. While we expect the quantum algorithm to deliver some speed-up at least in one of the input parameters, sometimes it might be worse than the best classical algorithm in some other parameters.

We mention here some of the fastest optimization algorithms in the quantum oracle model. In this paragraph we use the notation $\zo{}$ to hide polylogarithmic factors in any of the arguments. For solving an SDP with $m$ constraints involving $n \times n$ matrices, van Apeldoorn and Gily\'en~\cite{vAG19c} gave an algorithm that runs in time $\zo[\big]{(\sqrt{m} + \frac{\sqrt{n}}{\gamma})s / \gamma^4}$ where $s$ is the row-sparsity of the input matrices and $\gamma = \eps/Rr$ is the additive error $\eps$ of the algorithm scaled down with the upper bounds $R$ and $r$ on the respective sizes of the primal and dual solutions. This result builds on the classical Arora-Kale framework \cite{AK16j} that runs in time $\zo{mns/\gamma^4 + ns/\gamma^7}$, which was first quantized by Brand\~ao and Svore \cite{BS17c}. In~\cite{LCW19c} Li, Chakrabarti and Wu gave an $\zo{\sqrt{n}/ \eps^4 + \sqrt{d}/ \eps^8}$ time quantum algorithm for the classification of $n$ data points in dimension $d$ with margin $\eps$. Their design is the quantization of the work of Clarkson, Hazan and Woodruff~\cite{CHW12j} that runs in time $\zo{(n + d) / \eps^2}$. The same paper contains  similar quadratic quantum improvements over the classical constructions of~\cite{CHW12j} for kernel based classification, minimum enclosing ball and $\ell^2$-margin SVM, as well as an $\zo{\sqrt{n}/ \eps^4}$ time algorithm for zero-sum games with $n \times n$ payoff matrices. A similar result for zero-sum games, but with a better dependence on the error parameter, was obtained by van Apeldoorn and Gily\'en~\cite{vAG19p} whose algorithm runs in time $\zo{\sqrt{n}/ \eps^3}$. Both quantum algorithms for zero-sum games are based on the classical work of Grigoriadis and Khachiyan~\cite{GK95j} whose complexity is $\zo{n/ \eps^2}$. Finally, there is a series of quantum algorithms \cite{Jor05j,GAW19c,vAGGdW20j,CCLW18p} for fast gradient computation, which are typically combined with classical first order methods such as gradient descent.

%% file: PreviousWork.tex
The previous work on approximate submodular minimization \cite{HK12j,Bac13j,CLSW17c} is based on the subgradient descent method applied to the \lv\ extension. We review these results below, as it will help to present our contributions in the next section. From now on, we restrict our attention to submodular functions $F$ with range $[-1,1]$ and we seek for a set $\bar{S}$ such that $F(\bar{S}) \leq \min_S F(S) + \eps$.

\paragraph{Subgradient descent in \cite{HK12j}.}
Submodular minimization can be translated into a convex optimization problem by considering the so-called \lv\ extension $f$. This makes it possible to apply standard gradient algorithms. Since $f$ is not differentiable, one can rely on the subgradient descent method that computes a sequence of iterates $x\super{t}$ converging to a minimum of $f$. At each step, the next iterate $x\super{t+1}$ is obtained by moving into the negative direction of a subgradient $g\super{t}$ at $x\super{t}$. In the case of submodular functions, there exists a natural choice for $g\super{t}$, sometimes called the \lv\ subgradient, that requires $\bo{n / \eps^2}$ steps to converge to an $\eps$-approximate of the minimum. Since the \lv\ subgradient can be computed in time $\bo{n \cdot \eo + n \log n}$, the complexity of this approach is $\so{n^2 / \eps^2 \cdot \eo}$ \cite{HK12j}. The question arises if it is possible to find algorithms that scale better than $n^2$, i.e. that are subquadratic in the dimension.

\paragraph{Stochastic subgradient descent in \cite{CLSW17c}.}
In the above method, the subgradient $g\super{t}$ can equally be replaced with a stochastic subgradient, that is a low-variance estimate $\wg\super{t}$ satisfying $\ex{\wg\super{t} | x\super{t}} = g\super{t}$. One possible choice for $\wg\super{t}$ is the \emph{subgradient direct estimate} $\hg\super{t}$ defined as $\hg\super{t} = \norm{g\super{t}}_1 \sgn(g\super{t}_i) \cdot \vec{1}_i$ where $i \in [n]$ is sampled with probability $\abs{g\super{t}_i}/\norm{g\super{t}}_1$. The $\ell_1$-norm of the \lv\ subgradient being small, this is a low-variance $1$-sparse estimate of $g\super{t}$. However, it is unknown how to sample $\hg\super{t}$ faster than $\bo{n \cdot \eo + n \log n}$. Thus, using $\hg\super{t}$ at each step of the descent would not be more efficient than using the actual subgradient $g\super{t}$. Instead, the approach suggested in \cite{CLSW17c} is to use $\wg\super{t} = \hg\super{t}$ only once every $T = n^{1/3}$ steps, and in between to use $\wg\super{t} = \wg\super{t-1} + \wdi\super{t}$ where $\wdi\super{t}$ is an estimate of the subgradient difference $d\super{t} = g\super{t} - g\super{t-1}$. The crucial result in \cite{CLSW17c} is to show how to construct $\wdi\super{t}$ in time only proportional to the sparsity of $x\super{t} - x\super{t-1}$. This process has to be reset every $T$ steps since the variance and the sparsity of $\wg\super{t}$ increase over time (and therefore the sparsity of $x\super{t} - x\super{t-1}$ too). The amortized cost per step for constructing $\wg\super{t}$ is $\so{n^{2/3} \cdot \eo}$, leading to an $\so{n^{5/3}/\eps^2 \cdot \eo}$ algorithm.

It seems to us that there is a slight error in the construction of \cite{CLSW17c} because the estimate $\wg\super{t} = \wg\super{t-1} + \wdi\super{t}$ is not a valid stochastic subgradient, that is $\ex{\wg\super{t} | x\super{t}} \neq g\super{t}$. Indeed, even if $\wg\super{t-1}$ is an unbiased estimate of $g\super{t-1}$ conditioned on $x\super{t-1}$, in general that is not true conditioned on $x\super{t}$ (see a counterexample in Appendix \ref{App:counter}), thus $\ex{\wg\super{t} | x\super{t}} = \ex{\wg\super{t-1} | x\super{t}} + g\super{t} - g\super{t-1} \neq g\super{t}$. Nonetheless, this problem can be easily solved by sampling a second estimate $\wwg\super{t-1}$ of $g\super{t-1}$ such that, when conditioned on $x\super{t-1}$, it becomes independent of $\wg\super{t-1}$. Then $\wg\super{t}$ is redefined as $\wg\super{t} = \wwg\super{t-1} + \wdi\super{t}$. In this case, $x\super{t}$ does not convey any information about $\wwg\super{t-1}$ when conditioned on $x\super{t-1}$, implying that $\ex{\wwg\super{t-1} | x\super{t}} = g\super{t-1}$. However, in order to make $\wwg\super{t-1}$ independent of $\wg\super{t-1}$, \emph{all the previous} estimates involved in the computation of $\wg\super{t-1}$ have to be resampled. Since the construction of $\wg\super{t}$ was decomposed into batches of $T$ steps, it means that there are between $1$ and $T$ estimates to resample at each step. Nevertheless, a straightforward analysis shows that asymptotically the time complexity remains as stated in \cite{CLSW17c}.

%% file: Contributions.tex
Our method also consists of minimizing the \lv\ extension of the submodular function under consideration by using the stochastic subgradient descent algorithm. We differ from \cite{HK12j,CLSW17c} by constructing a new subgradient oracle that is faster to evaluate. In the quantum model, our further speed-up is based on two new results that might be of independent interest. One is a simple proof of robustness for the (classical) stochastic subgradient descent method when the subgradient oracle has some biased noise. The other one is a new quantum algorithm for sampling multiple independent elements from discrete probability distributions. These results are detailed below.

\paragraph{Classical algorithm.} Similarly to \cite{CLSW17c}, we construct our subgradient oracle $\wg\super{t}$ by combining two kinds of estimates (Section \ref{Sec:frame}). Our construction is reset every $T = n^{1/2}$ steps, which turns out to be the optimal resetting time in our case. We explain how to compute the first $T$ terms $\wg\super{0}, \dots, \allowbreak \wg\super{T-1}$. First, we obtain $T$ independent samples $\hg\super{0,0}, \cdots, \hg\super{0,T-1}$ from the subgradient direct estimate at $x\super{0}$. Using a standard sampling method (Lemma \ref{Lem:samplClass}), this can be done in time $\so{n \cdot \eo + T}$ (Proposition \ref{Prop:samplgCl}). Then, the first subgradient estimate is chosen to be $\wg\super{0} = \hg\super{0,0}$, and the other ones are obtained at step $t$ by combining $\hg\super{0,t}$ with an estimate $\wdi\super{t}$ of the \lv\ subgradient difference $d\super{t} = g\super{t} - g\super{0}$, that is $\wg\super{t} = \hg\super{0,t} + \wdi\super{t}$. Notice that the difference is not taken between two consecutive iterates $x\super{t-1}$ and $x\super{t}$ as in \cite{CLSW17c}, but between the first iterate $x\super{0}$ and the current one $x\super{t}$. This has the advantage of keeping the variance under control since we add up two terms instead of $t+1$. Moreover, the sparsity increases only linearly, instead of quadratically, in $t$. Our procedure for constructing $\wdi\super{t}$ is directly adapted from \cite{CLSW17c}, with time complexity $\so{t \cdot \eo}$ (Proposition \ref{Prop:sampldCl}). Consequently, the first $T$ estimates are obtained in time $\so[\big]{(n \cdot \eo + T) + \sum_{t=1}^{T-1} t \cdot \eo} = \so{n \cdot \eo}$. Since the $\bo{n / \eps^2}$ steps of the subgradient descent are split into $\bo{\sqrt{n} / \eps^2}$ batches of length $T$, it follows that the total time complexity is $\so{n^{3/2}/\eps^2 \cdot \eo}$.

\vspace*{10pt}
{\noindent \bf Statement of Theorem \ref{Thm:clSub}.} {\itshape There is a classical algorithm that, given a submodular function $F : 2^V \ra [-1,1]$ and $\eps > 0$, computes a set $\bar{S}$ such that $\ex{{F}(\bar{S})} \leq \min_{S \subseteq V} F(S) + \eps$ in time $\so{n^{3/2} /\eps^{2} \cdot \eo}$.}

\paragraph{Quantum algorithm.} We first note that there is a simple $\so{n^{3/2}/\eps^3 \cdot \eo}$ quantum algorithm using only the subgradient direct estimate $\hg\super{t}$. The latter was defined as $\hg\super{t} = \norm{g\super{t}}_1 \sgn(g\super{t}_i) \cdot \vec{1}_i$ where $i \in [n]$ is sampled from the probability distribution $\bigl(\abs{g\super{t}_1}/\norm{g\super{t}}_1, \allowbreak\dots, \abs{g\super{t}_n}/\norm{g\super{t}}_1\bigr)$. It is a standard result that one sample from any discrete probability distribution $(p_1, \dots, p_n)$ (given as an evaluation oracle) can be obtained in time $\bo{\sqrt{n \cdot \max_i p_i} \cdot \eo} = \bo{\sqrt{n} \cdot \eo}$ by quantum state preparation of $\sum_{i \in [n]} \sqrt{p_i} \ket{i}$ (Lemma \ref{Lem:samplOne}). Moreover, the $\ell_1$-norm of any $n$-coordinates vector can be estimated with accuracy $\eps$ in time $\bo{\sqrt{n}/\eps \cdot \eo}$ using the Amplitude Estimation algorithm (Lemma \ref{Lem:amplEst}). A straightforward combination of these two results leads to an $\eps$-biased estimate $\hg_{\eps}\super{t}$, satisfying $\norm{\ex{\hg_{\eps}\super{t} | x\super{t}} - g\super{t}}_1 \leq \eps$, that can be computed in time $\so{\sqrt{n}/\eps \cdot \eo}$. This does not meet the usual requirement of an unbiased estimate for the stochastic subgradient descent method. However, we prove that the latter is robust to such a noise (Proposition \ref{Prop:gradDescent}). This leads to an $\so{n^{3/2}/\eps^3 \cdot \eo}$ algorithm for approximate submodular minimization.

We now describe our quantum algorithm achieving time complexity $\so{n^{5/4}/\eps^{5/2} \cdot \log(1/\eps) \cdot \eo}$, based on our enhanced classical algorithm. Similarly to the simple result described above, we accelerate the construction of the subgradient estimate $\wg\super{t}$ using quantum sampling. A first attempt would be to apply the quantum state preparation method to sample each estimate individually. However, the computation of the first $T$ estimates of $g\super{0}$, for instance, would incur a cost of $\so{T \sqrt{n}/\eps \cdot \eo}$, which is worse than classically when $T = n^{1/2}$ (we could change the value of $T$ but that does not improve the overall complexity). We overcome this issue by using a new quantum multi-sampling algorithm for sampling $T$ independent elements from any discrete probability distribution $(p_1, \dots, p_n)$ in time $\bo{\sqrt{Tn} \cdot \eo}$, instead of $\bo{T\sqrt{n} \cdot \eo}$. This algorithm is described in the next paragraph. It leads to our second main result.

\vspace*{10pt}
{\noindent \bf Statement of Theorem \ref{Thm:quSub}.} {\itshape There is a quantum algorithm that, given a submodular function $F : 2^V \ra [-1,1]$ and $\eps > 0$, computes a set $\bar{S}$ such that $\ex{{F}(\bar{S})} \leq \min_{S \subseteq V} F(S) + \eps$ in time $\so{n^{5/4} /\eps^{5/2} \cdot \log(\frac{1}{\eps}) \cdot \eo}$.}

\paragraph{Quantum multi-sampling algorithm.} We now sketch the algorithm for sampling $T$ elements from $p = (p_1, \dots, p_n)$ in time $\bo{\sqrt{Tn} \cdot \eo}$ (here $\eo$ is the time complexity of a quantum evaluation oracle $\mathcal{O}_p$ satisfying $\mathcal{O}_p(\ket{i} \ket{0}) = \ket{i} \ket{p_i}$ for all $i$). First, we find the set $S$ of all the coordinates $i \in [n]$ where $p_i$ is larger than $1/T$. Since there are at most $T$ values to find, $S$ can be computed in time $\bo{\sqrt{Tn} \cdot \eo}$ using Grover search. Then, we load in time $\bo{T \cdot \eo}$ the conditional distribution $(p_i/p_S)_{i \in S}$, where $p_S = \sum_{i \in S} p_i$, into a classical data structure \cite{Vos91j} that supports fast sampling from $(p_i/p_S)_{i \in S}$ in time $\bo{1}$. On the other hand, we can sample from the complement distribution $(p_i/(1-p_S))_{i \notin S}$ using quantum state preparation of $\frac{1}{\sqrt{1-p_S}} \sum_{i \notin S} \sqrt{p_i} \ket{i}$ in time $\bo{\sqrt{n \cdot \max_{i \notin S} p_i/(1-p_S)} \cdot \eo} = \bo{\sqrt{n/(T(1-p_S))} \cdot \eo}$. Now, each of the $T$ samples is obtained by first flipping a coin that lands head with probability $p_S$, and then sampling $i \in S$ from the classical data structure (head case) or $i \notin S$ by quantum state preparation (tail case). The total expected time is $\bo[\big]{T p_S \cdot 1 + T(1-p_S) \cdot \sqrt{n/(T(1-p_S))} \cdot \eo} = \bo{\sqrt{Tn} \cdot \eo}$ (assuming $T < n$). Additional technicalities, arising from the fact that $(p_1, \dots, p_n) = (u_1/\norm{u}_1,\dots,u_n/\norm{u}_1)$ may be given as an unormalized vector $(u_1,\dots,u_n)$, are also discussed in this paper (Section \ref{Sec:Sampling}).

\vspace*{10pt}
{\noindent \bf Statement of Theorem \ref{Thm:setupEasy}.} {\itshape There is a quantum algorithm that, given an integer $1 < T < n$, a real $0 < \delta < 1$, and an evaluation oracle to a discrete probability distribution $\dis{} = (p_1, \dots, p_n)$, outputs $T$ independent samples from $\dis{}$ in expected time $\bo[\big]{\sqrt{Tn} \log(1/\delta) \cdot \eo}$ with probability $1-\delta$.}

%% file: Preliminaries.tex
\paragraph{Notations.}
Let $[n] = \{1,\dots,n\}$. Given a vector $u \in \R^n$ and a positive integer $p$, we let $\norm{u}_p = \bigl(\sum_{i \in [n]} |u_i|^p\bigr)^{1/p}$ be the $\ell_p$-norm of $u$, and $\norm{u}_{\infty} = \max_{i \in [n]} |u_i|$ be the largest entry (in absolute value). We say that $u$ is $k$-sparse if it has at most $k$ non-zero entries. We denote by $u\ps \in \R^n$ (resp. $u\ms \in \R^n$) the vector obtained from $u$ by replacing its negative (resp. positive) entries with $0$ (thus, $u = u\ps + u\ms$). Given two vectors $u, u' \in \R^n$, we use $u \geq u'$ (resp. $u \leq u'$) to denote that $u - u' \in \R^n\ps$ (resp. $u - u' \in \R^n\ms$). We also let $\sgn(u)$ to be $1$ if $u \geq 0$, and $-1$ otherwise. Given a set $S \subseteq [n]$, we denote by $u_S \in \R^{|S|}$ the subvector $(u_i)_{i \in S}$ of $u$ made of the values at coordinates $i \in S$. If $u$ is a non-zero vector, we define $\dis{u}$ to be the probability distribution $\left(\frac{\abs{u_1}}{\norm{u}_1},\dots,\frac{\abs{u_n}}{\norm{u}_1}\right)$ on $[n]$. Finally, we let $\vec{1}_i \in \R^n$ be the indicator vector with a $1$ at position $i \in [n]$ and $0$ elsewhere.

\paragraph{\lv\ extension.}
A submodular function $F$ is a set function $F : 2^V \ra \R$, over some ground set $V$ of size $n$, that satisfies the diminishing returns property: for every $A \subseteq B \subseteq V$ and for every $i \not\in B$, the inequality $F(A \cup \{i\}) - F(A) \geq F(B \cup \{i\}) - F(B)$ holds. For convenience, and without loss of generality, we assume that $V = [n]$ and $F(\varnothing) = 0$ (this can be enforced by observing that $S \mapsto F(S) - F(\varnothing)$ is still a submodular function). The \lv\ extension $f : [0,1]^n \ra \R$ is a convex relaxation of $F$ to the hypercube $[0,1]^n$. Before describing it, we present a canonical way to associate a permutation $P$ with each $x \in [0,1]^n$.

\begin{definition}
  Given a permutation $P = (P_1, \dots, P_n)$ of $[n]$, we say that $P$ is \emph{consistent} with $x \in \R^n$ if $x_{P_1} \geq x_{P_2} \geq \dots \geq x_{P_n}$, and $P_{i+1} > P_i$ when $x_{P_i} = x_{P_{i+1}}$ for all $i$. We also denote $P[i] = \{P_1,\dots,P_i\} \subseteq [n]$ the set of the first $i$ elements of $P$, and $P[0] = \varnothing$.
\end{definition}

As an example, the permutation $P$ consistent with $x = (0.3,0.2,0.3,0.1)$ is $P = (1,3,2,4)$.

\begin{definition}
  Given a submodular function $F : 2^V \ra \R$ over $V = [n]$, the \emph{\lv\ extension} $f : [0,1]^n \ra \R$ of $F$ is defined for all $x \in [0,1]^n$ by $f(x) = \sum_{i \in [n]} (F(P[i]) - F(P[i-1])) \cdot x_{P_i}$ where $P$ is the permutation consistent with $x$. The \emph{\lv\ subgradient} $g(x) \in \R^n$ at $x \in [0,1]^n$ is defined by $g(x)_{P_i} = F(P[i]) - F(P[i-1])$ for all $i \in [n]$.
\end{definition}

The following standard properties of the \lv\ extension \cite{Lov82c,Bac13j,JB11c} will be used in this paper.

\begin{proposition}
  \label{Prop:lv}
  The \lv\ extension $f$ of a submodular function $F$ is a convex function. Moreover, given $x \in [0,1]^n$ and the permutation $P$ consistent with $x$, we have
  \begin{enumerate}
    \item \emph{\textbf{(Subgradient)}} For all $y \in [0,1]^n$, $\inp{g(x)}{x-y} \geq f(x) - f(y)$.
    \item \emph{\textbf{(Minimizers)}} $\min_{i \in [n]} F(P[i]) \leq f(x)$ and $\min_{S \subseteq V} F(S) = \min_{y \in [0,1]^n} f(y)$.
    \item \emph{\textbf{(Boundedness)}} If the range of $F$ is $[-1,1]$ then $\norm{g(x)}_2 \leq \norm{g(x)}_1 \leq 3$.
  \end{enumerate}
\end{proposition}

Observe that the second property gives an explicit way to convert any $\bar{x} \in [0,1]^n$ such that $f(\bar{x}) \leq \min_{x \in [0,1]^n} f(x) + \eps$ into a set $\bar{S} \subseteq V$ such that $F(\bar{S}) \leq \min_{S \subseteq V} F(S) + \eps$. Consequently, we can focus on $\eps$-additive minimization of the \lv\ extension in the rest of the paper.

\subparagraph{Models of Computation.}
We describe the two models of computation used in this paper. Although the \lv\ extension is a continuous function, given $x \in [0,1]^n$ it is sufficient to evaluate $F$ on the sets $P[1], \dots, P[n]$ to compute $f(x)$, where $P$ is the permutation consistent with $x$. The same holds for the \lv\ subgradient. Consequently, given $P$, it is natural to define an evaluation oracle that given $i$ returns $F(P[i])$. The input $i$ to this oracle is encoded over $\bo{\log n}$ bits, whereas representing each of the sets $P[i]$ as an indicator vector over $\rn^n$ would require $n$ bits.
  \begin{itemize}
    \item \emph{Classical Model.} We use the same model as described in \cite{CLSW17c}. The submodular function $F$ can be accessed via an evaluation oracle that takes as input an integer $i \in [n]$ and a linked list storing a permutation $P$ of $[n]$, and returns the value of $F(P[i])$. We denote by $\eo$ the cost of one evaluation query to the oracle.
    \item \emph{Quantum Model.} We extend the above model to the quantum setting in a standard way. Given a permutation $P$ of $[n]$ stored in a linked list, we assume that we have access to a unitary operator $\mathcal{O}_P$ that, given $i \in [n]$, satisfies $\mathcal{O}_P(\ket{i} \ket{0}) = \ket{i} \ket{F(P[i])}$, where the second register holds a binary representation of $F(P[i])$ with some finite precision. We denote by $\eo$ the cost of one evaluation query to $\mathcal{O}_P$.
  \end{itemize}
The \lv\ extension $f(x)$ at $x$ can be evaluated in time $\bo{n \log n + n \cdot \eo}$ in the above models.

%% file: Sampling.tex
We study the problem of generating $T$ independent samples from a discrete probability distribution $\dis{u} = \bigl(\frac{\abs{u_1}}{\norm{u}_1},\dots,\frac{\abs{u_n}}{\norm{u}_1}\bigr)$ on $[n]$, where $u = (u_1,\dots,u_n) \in \R^n$ is a non-zero vector given as an evaluation oracle. This task is a fundamental part of Monte Carlo methods and discrete events simulation \cite{Dev86b,BFS87b}. Here, it will be used to construct randomized estimators of the \lv\ subgradient in Sections \ref{Sec:classApprox} and \ref{Sec:quantApprox}. In this section, $\eo$ denotes the time complexity of an evaluation oracle to $u$. In the classical setting, this oracle must return $u_i$ given $i \in [n]$, whereas in the quantum setting it is a unitary operator $\mathcal{O}_u$ satisfying $\mathcal{O}_u(\ket{i} \ket{0}) = \ket{i} \ket{u_i}$ for all $i$.

The above problem has been thoroughly investigated in the classical setting \cite{Dev86b,BFS87b}, where it can be solved in time $\bo{n \cdot \eo + T}$ using the alias method. We present this result below, as it will be part of our quantum algorithm later.

\begin{lemma}[\cite{Wal74j,Vos91j}]
  \label{Lem:samplClass}
  There is a classical algorithm that, given an evaluation oracle to a non-zero vector $u \in \R^n$, constructs in time $\bo{n \cdot \eo}$ a data structure from which one can output as many independent samples from $\dis{u}$ as desired, each in time $\bo{1}$.
\end{lemma}

In the quantum setting, it is a well-known result that one sample from $\dis{u}$ can be obtained by preparing the state $\sum_{i \in [n]} \sqrt{\frac{\abs{u_i}}{\norm{u}_1}} \ket{i}$ with Amplitude Amplification and measuring the $\ket{i}$ register.

\begin{lemma}[\cite{Gro00j}]
  \label{Lem:samplOne}
  There is a quantum algorithm that, given an evaluation oracle to a non-zero vector $u \in \R^n$ and a value $M \geq \norm{u}_{\infty}$, outputs one sample from $\dis{u}$ in expected time $\bo[\Big]{\sqrt{\frac{n M}{\norm{u}_1}} \cdot \eo}$.
\end{lemma}

Note that the maximum $M = \norm{u}_{\infty}$ of any vector $u \in \R^n$ can be computed with high probability using D\"urr-H{\o}yer's algorithm \cite{DH96p} in time $\bo{\sqrt{n} \cdot \eo}$, in which case we have $\sqrt{nM/\norm{u}_1} \leq \sqrt{n}$. Then, by simply repeating the above algorithm $T$ times, one can obtain $T$ samples in time $\bo{T \sqrt{n} \cdot \eo}$. Our main contribution (Algorithm \ref{Algo:KSampl}) is to improve this time complexity to $\bo{\sqrt{Tn} \cdot \eo}$. If the normalization factor $\norm{u}_1$ is unknown, we will only be able to sample from a distribution $\dis{u}(\Gamma,S)$ close to $\dis{u}$ that is defined below. Here, $\Gamma > 0$ acts as a placeholder for an estimate of $\norm{u}_1$ and $S \subseteq [n]$ is meant to contain the indices $i$ where $\abs{u_i}$ is larger than $\Gamma/T$.

\begin{definition}
  Consider a non-zero vector $u \in \R^n$. Fix a real number $\Gamma > 0$ and a set $S \subseteq [n]$ such that $\Gamma \geq \norm{u_S}_1$. We define $\dis{u}(\Gamma,S)$ to be the distribution that outputs $i \in [n]$ with probability
  \[
    \left\{
        \begin{array}{ll}
            \frac{|u_i|}{\Gamma} & \mbox{if } i \in S \\ [3pt]
            \frac{|u_i|}{\Gamma} + \left(1 - \frac{\norm{u}_1}{\Gamma}\right) \frac{|u_i|}{\norm{u_{[n] \setminus S}}_1} = \left(1 - \frac{\norm{u_S}_1}{\Gamma}\right) \frac{|u_i|}{\norm{u_{[n] \setminus S}}_1} & \mbox{if $i \in [n] \setminus S$.}
        \end{array}
    \right.
  \]
  Note that if $\Gamma = \norm{u}_1$ then $\dis{u}(\norm{u}_1,S) = \dis{u}$, which is independent of $S$.
\end{definition}

We now prove that Algorithm \ref{Algo:KSampl} runs in time $\bo{\sqrt{Tn} \cdot \eo}$ when $\Gamma$ is sufficiently close to $\norm{u}_1$ and $S = \{i \in [n] : |u_i| \geq \Gamma/T\}$. We will explain later how to find such parameters in time $\bo{\sqrt{Tn} \cdot \eo}$.

\begin{algorithm}[htbp]
\caption{Sampling $T$ elements from $\dis{u}(\Gamma,S)$.}
\label{Algo:KSampl}
\normalsize
\textbf{Input:} a non-zero vector $u \in \R^n$, an integer $1 < T < n$, a real $\Gamma > 0$ and a set $S \subseteq [n]$ such that $\Gamma \geq \norm{u_S}_1$, the value $M = \norm{u_{[n] \setminus S}}_{\infty}$. \\
\textbf{Output:} a sequence $(i_1, \dots, i_T) \in [n]^T$. \vspace{0.1cm}

\begin{algorithmic}[1]
\State Construct the data structure associated with $u_S = (u_i)_{i \in S}$ in Lemma \ref{Lem:samplClass}, and compute $\norm{u_S}_1$.
\For{$t = 1, \dots, T$}
  \State Sample $b_t \in \{0,1\}$ from the Bernoulli distribution of parameter $p = \frac{\norm{u_S}_1}{\Gamma}$.
  \State If $b_t = 1$, sample $i_t \sim \dis{u_S}$ using the data structure built at step 1.
  \State If $b_t = 0$, sample $i_t \sim \dis{u_{[n] \setminus S}}$ using Lemma \ref{Lem:samplOne} with input $u_{[n] \setminus S}$ and $M$.
\EndFor
\State Output $(i_1, \dots, i_T)$.
\end{algorithmic}
\end{algorithm}

\begin{theorem}
  \label{Thm:KSampl}
  The output $(i_1, \dots, i_T) \in [n]^T$ of Algorithm \ref{Algo:KSampl} consists of $T$ independent samples from the distribution $\dis{u}(\Gamma,S)$. Moreover, if $|\Gamma - \norm{u}_1| \leq \norm{u}_1/\sqrt{T}$ and $S = \{i \in [n] : |u_i| \geq \Gamma/T\}$ then the expected run-time of the algorithm is $\bo{\sqrt{Tn} \cdot \eo}$.
\end{theorem}

\begin{proof}
  \setlength{\abovedisplayskip}{2.5pt}
  \setlength{\belowdisplayskip}{2.5pt}
  At each execution of lines 2-5, the probability to sample $i \in S$ is $\frac{\norm{u_S}_1}{\Gamma} \cdot \frac{|u_i|}{\norm{u_S}_1} = \frac{|u_i|}{\Gamma}$ and the probability to sample $i \in [n] \setminus S$ is $\bigl(1 - \frac{\norm{u_S}_1}{\Gamma}\bigr) \frac{|u_i|}{\norm{u_{[n] \setminus S}}_1}$. This is the distribution $\dis{u}(\Gamma,S)$.

  We now analyze the time complexity. Line 1 takes time $\bo{|S| \cdot \eo}$. Each execution of line 4 takes time $\bo{1}$, and each execution of line 5 takes time $\bo[\Big]{\sqrt{n \cdot \norm{u_{[n] \setminus S}}_{\infty}/\norm{u_{[n] \setminus S}}_1} \cdot \eo}$ (according to Lemma \ref{Lem:samplOne}). Thus, the expected run-time of the algorithm is
    \[\bo*{|S| \cdot \eo + T \frac{\norm{u_S}_1}{\Gamma} \cdot 1 + T \left(1 - \frac{\norm{u_S}_1}{\Gamma}\right) \cdot \sqrt{\frac{n \cdot \norm{u_{[n] \setminus S}}_{\infty}}{\norm{u_{[n] \setminus S}}_1}} \cdot \eo}.\]
  Assume that $|\Gamma - \norm{u}_1| \leq \norm{u}_1/\sqrt{T}$ and $S = \{i \in [n] : |u_i| \geq \Gamma/T\}$. Since $T \geq 2$, it follows that $\Gamma \geq \norm{u}_1/4$ and $|S| \leq 4T$. Consequently, $1 - \frac{\norm{u_S}_1}{\Gamma} \leq \bigl(1 + \frac{1}{\sqrt{T}}\bigr) \frac{\norm{u}_1}{\Gamma} - \frac{\norm{u_S}_1}{\Gamma} \leq \frac{\norm{u_{[n] \setminus S}}_1}{\Gamma} + \frac{4}{\sqrt{T}}$. Moreover, $\norm{u_{[n] \setminus S}}_{\infty} \leq \min(\Gamma/T,\norm{u_{[n] \setminus S}}_1)$. Thus, the expected run-time is
    \[\bo[\Bigg]{T \cdot \eo + T \left(\frac{\norm{u_{[n] \setminus S}}_1}{\Gamma} + \frac{1}{\sqrt{T}}\right) \cdot \sqrt{\frac{n \cdot \min(\Gamma/T,\norm{u_{[n] \setminus S}}_1)}{\norm{u_{[n] \setminus S}}_1}} \cdot \eo} = \bo*{\sqrt{Tn} \cdot \eo}.\]
\end{proof}

The above result is optimal, as can be shown by a simple reduction from the $T$-search problem. We now explain how to find the values $\Gamma$, $S$ and $\norm{u_{[n] \setminus S}}_{\infty}$ needed by Algorithm \ref{Algo:KSampl}. First, if $\norm{u}_1$ is known, we can assume without loss of generality that $\norm{u}_1 = 1$. In this case, we obtain $T$ samples from $\dis{u} = (p_1, \dots, p_n)$ as follows.

\begin{theorem}
  \label{Thm:setupEasy}
  There is a quantum algorithm that, given an integer $1 < T < n$, a real $0 < \delta < 1$, and an evaluation oracle to a discrete probability distribution $\dis{} = (p_1, \dots, p_n)$, outputs $T$ independent samples from $\dis{}$ in expected time $\bo[\big]{\sqrt{Tn} \log(1/\delta) \cdot \eo}$ with probability $1-\delta$.
\end{theorem}

\begin{proof}
  The set $S = \{i \in [n] : |p_i| \geq 1/T\}$ and the value $M = \norm{p_{[n] \setminus S}}_{\infty}$ can be computed with  probability $1-\delta$ using Grover search and D\"urr-H{\o}yer's algorithm \cite{DH96p} in time $\bo{\sqrt{Tn} \log(1/\delta) \cdot \eo}$ and $\bo{\sqrt{n} \log(1/\delta) \cdot \eo}$ respectively. Then, conditioned on these two values to be correct, Algorithm \ref{Algo:KSampl} outputs $T$ independent samples from $\dis{}$ in expected time $\bo{\sqrt{Tn} \cdot \eo}$ (where we use $\Gamma = 1$).
\end{proof}

If $\norm{u}_1$ is unknown (as it will be the case in our applications), we will need the next result about Amplitude Estimation \cite{BHMT02j} to approximate its value.

\begin{lemma}
  \label{Lem:amplEst}
  There is a quantum algorithm that, given an evaluation oracle to a non-zero vector $u \in \R^n$, a value $M \geq \norm{u}_{\infty}$ and two reals $0 < \eps, \delta < 1$, outputs a real $\Gamma$ such that $\abs{\Gamma - \norm{u}_1} \leq \eps \norm{u}_1$ with probability $1-\delta$. The expected run-time of this algorithm is $\bo*{\frac{1}{\eps}\sqrt{\frac{nM}{\norm{u}_1}} \log(1/\delta) \cdot \eo}$.
\end{lemma}

\begin{proof}
  Define $V_{u,M}$ to be a unitary operator such that
    \[V_{u,M} (\ket{0}\ket{0})
       = \frac{1}{\sqrt{n}} \sum_{i \in [n]} \ket{i} \left(\sqrt{\frac{|u_i|}{M}} \ket{0} + \sqrt{1 - \frac{|u_i|}{M}} \ket{1}\right)
       = \sqrt{\frac{\norm{u}_1}{n M}} \ket{\psi_u} \ket{0} + \sqrt{1 - \frac{\norm{u}_1}{n M}} \ket{\phi_u} \ket{1}
    \]
  where $\ket{\psi_u} = \sum_{i \in [n]} \sqrt{\frac{|u_i|}{\norm{u}_1}} \ket{i}$, and $\ket{\phi_u}$ is some unit vector.   $V_{u,M}$ can be constructed with two quantum queries to $u$ and a controlled rotation (see also \cite{SLSB19j} for an alternative construction). Now, using the Amplitude Estimation algorithm \cite[Theorem 12]{BHMT02j} on $V_{u,M}$ with accuracy $\eps$, we get an estimate $\gamma$ such that $\abs{\gamma - \norm{u}_1/(n M)} \leq \eps \norm{u}_1/(n M)$ with probability $2/3$ in expected time $\bo*{\frac{1}{\eps}\sqrt{\frac{nM}{\norm{u}_1}} \cdot \eo}$. The success probability can be increased to $1-\delta$ by a standard Chernoff bound argument at an extra cost factor $\log(1/\delta)$. Finally, we take $\Gamma = nM\gamma$.
\end{proof}

The construction of the setup parameters $(\Gamma,S,M)$ is described in Algorithm \ref{Algo:setup}. We need to be careful that $\Gamma \geq \norm{u_S}_1$, otherwise $\dis{u}(\Gamma,S)$ is not a probability distribution. The parameter $\eps$ controls the closeness of $\dis{u}(\Gamma,S)$ to $\dis{u}$. We have $\eps' = \min(1/\sqrt{T},\eps)$ to guarantee that $|\Gamma - \norm{u}_1| \leq (1/\sqrt{T}) \norm{u}_1$. The setup cost is dominated by $\bo{\sqrt{n}/\eps}$ if  $\eps \leq 1/\sqrt{T}$.

\begin{algorithm}[htbp]
\caption{Construction of the setup parameters $(\Gamma,S,M)$.}
\label{Algo:setup}
\normalsize
\textbf{Input:} a non-zero vector $u \in \R^n$, an integer $1 < T < n$, two reals $0 < \eps, \delta < 1$. \\
\textbf{Output:} a real $\Gamma$, a set $S \subseteq [n]$, a value $M$. \vspace{0.1cm}

The subroutines below are run with failure parameter $\delta/4$. The algorithm aborts and outputs $\emph{fail}$ if any step takes time greater than $c \cdot (\sqrt{Tn} + \sqrt{n}/\eps) \log(1/\delta)$ (where $c$ is a constant to be specified in the proof of Proposition \ref{Prop:setupShort}). \vspace{0.1cm}

\begin{algorithmic}[1]
\State Run D\"urr-H{\o}yer's algorithm \cite{DH96p} to compute $\norm{u}_{\infty}$. Denote the result by $L$.
\State Compute an estimate $\hat{\Gamma}$ of $\norm{u}_1$ with relative error $\eps' = \min(1/\sqrt{T},\eps)$ using $L$ and Lemma \ref{Lem:amplEst}.
\State Run the Grover search algorithm \cite{BBHT98j} on $u$ to find all the indices $i$ such that $|u_i| \geq \hat{\Gamma}/T$. Denote the result by $\hat{S} \subseteq [n]$.
\State Compute $\norm{u_{\hat{S}}}_1$ and set $\Gamma = \max\{\norm{u_{\hat{S}}}_1, \hat{\Gamma}\}$. Compute $S = \{i \in \hat{S} : |u_i| \geq \Gamma/T\}$.
\State Run D\"urr-H{\o}yer's algorithm \cite{DH96p} to compute $\norm{u_{[n] \setminus S}}_{\infty}$. Denote the result by $M$.
\State Output $(\Gamma,S,M)$.
\end{algorithmic}
\end{algorithm}

\begin{proposition}
  \label{Prop:setupShort}
  The output $(\Gamma,S,M)$ of Algorithm \ref{Algo:setup} satisfies $\Gamma \geq \norm{u_S}_1$, $|\Gamma - \norm{u}_1| \leq \min(1/\sqrt{T}, \allowbreak \eps) \norm{u}_1$, $S = \{i \in [n] : |u_i| \geq \Gamma/T\}$ and $M = \norm{u_{[n] \setminus S}}_{\infty}$ with probability $1-\delta$. The expected run-time of this algorithm is $\bo[\big]{(\sqrt{Tn} + \sqrt{n}/\eps) \log(1/\delta) \cdot \eo}$.
\end{proposition}

\begin{proof}
  We first assume that all steps of the algorithm succeed and do not abort. In this case, we have $|\hat{\Gamma} - \norm{u}_1| \leq \eps' \norm{u}_1$. Thus, $\Gamma = \max\{\norm{u_{\hat{S}}}_1, \hat{\Gamma}\} \geq \max\{\norm{u_S}_1, (1-\eps') \norm{u}_1\}$ and $\Gamma \leq (1+\eps') \norm{u}_1$. Moreover, $S = \{i \in [n] : |u_i| \geq \Gamma/T\}$ since $\{i \in [n] : |u_i| \geq \Gamma/T\} \subseteq \{i \in [n] : |u_i| \geq \hat{\Gamma}/T\} = \hat{S}$.

  We now study the time needed by lines 1-5 to succeed with probability $1-\delta$. If we omit the $\log(1/\delta) \cdot \eo$ factors, then there exist four absolute constants $c_1$, $c_2$, $c_3$ and $c_4$ such that lines 1 and 5 need time $c_1 \cdot \sqrt{n}$, line 2 needs time $c_2 \cdot \frac{1}{\eps'}\sqrt{n L / \norm{u}_1} \leq c_2 \cdot (\sqrt{Tn} + \sqrt{n}/\eps)$ (according to Lemma \ref{Lem:amplEst}, and since $L = \norm{u}_{\infty} \leq \norm{u}_1$ if line 1 succeeds), line 3 needs time $c_3 \cdot \sqrt{Tn}$ (since $\hat{S} \leq 4T$ if $\hat{S} = \{i \in [n] : |u_i| \geq \hat{\Gamma}/T\}$ and $\hat{\Gamma} \geq (1-\eps') \norm{u}_1 \geq \norm{u}_1/4$) and line 4 needs time $c_4 \abs{\hat{S}} \leq 4c_4 T$. Consequently, if we take $c = \max\{c_1, c_2, c_3, 4c_4\}$, the algorithm does not abort and succeeds with probability $1-\delta$.
\end{proof}

%% file: Framework.tex
In this section, we construct our new low-variance estimate of the \lv\ subgradient, and we apply the stochastic subgradient descent algorithm on it to minimize the \lv\ extension. The stochastic subgradient descent method is a general algorithm for approximating the minimum value of a convex function $f$ that is not necessarily differentiable (as it is the case for the \lv\ extension). It uses the concept of \emph{subgradients} (or \emph{subderivatives}) of $f$, which is defined as follows.

\begin{definition}
  \label{Def:subg}
  Given a convex function $f : C \ra \R$ over $C \subset \R^n$ and a point $x \in C$, we say that $g \in \R^n$ is a \emph{subgradient} of $f$ at $x$ if $\inp{g}{x-y} \geq f(x) - f(y)$ for all $y \in C$. The set of all subgradients at $x$ is denoted by $ \partial f(x)$.
\end{definition}

Normally, the stochastic subgradient descent method requires to compute a sequence $(\wg\super{t})_t$ of \emph{unbiased} subgradient estimates at certain points $(x\super{t})_t$, which means that $\ex{\wg\super{t} | x\super{t}} \in \partial f(x\super{t})$. In the next proposition, we generalize this method to $\eps$-noisy estimates satisfying only $\norm{\ex{\wg\super{t} | x\super{t}} - g\super{t}}_1 \leq \eps$ for some $g\super{t} \in \partial f(x\super{t})$. In the case $\eps = 0$, our analysis recovers the standard error bound \cite{Duc18b}.

\begin{proposition}
  \label{Prop:gradDescent}
  \emph{\textbf{(Noisy Stochastic Subgradient Descent)}}
  Let $f : C \ra \R$ be a convex function over a compact convex set $C \subset \R^n$, and $\eta > 0$. Consider two sequences of random variables $(x\super{t})_t$ and $(\wg\super{t})_t$ such that $x\super{0} = \argmin_{x \in C} \norm{x}_2$, $x\super{t+1} = \argmin_{x \in C} \norm{x - (x\super{t} - \eta \wg\super{t})}_2$, and
  \begin{center}
    $\norm*{\ex*{\wg\super{t} | x\super{t}} - g\super{t}}_1 \leq \eps$ for some $g\super{t} \in \partial f(x\super{t}),$
  \end{center}
  for all $t \geq 0$. Fix $\xopt \in \argmin_{x \in C} f(x)$ and let $L_2, L_{\infty}, B \in \R$ be such that $\norm{x - \xopt}_2 \leq L_2$, $\norm{x - \xopt}_{\infty} \leq L_{\infty}$ and $\ex*{\norm{\wg\super{t}}_2^2} \leq B^2$, for all $x \in C$ and $t \geq 0$. Then, for any integer $N$, the average point $\bar{x} = \frac{1}{N} \sum_{t=0}^{N-1} x\super{t}$ satisfies
    $\ex*{f(\bar{x})} \leq f(\xopt) + \frac{L_2^2}{2\eta N} + \frac{\eta}{2} B^2 + \eps L_{\infty}$.
\end{proposition}

\begin{proof}
  Let $(g\super{t})_t$ be such that $g\super{t} \in \partial f(x\super{t})$ and $\norm{\ex{\wg\super{t} | x\super{t}} - g\super{t}}_1 \leq \eps$. Then,
    \begin{align*}
      \norm{x\super{t+1} - \xopt}_2^2
          & = \norm[\Big]{\argmin_{x \in C} \norm{x - (x\super{t} - \eta \wg\super{t})}_2 - \xopt}_2^2 \\
          & \leq \norm{x\super{t} - \eta \wg\super{t} - \xopt}_2^2 \quad \mbox{by property of the projection onto $C$}\\
          & = \norm{x\super{t} - \xopt}_2^2 - 2 \eta \inp{\wg\super{t}}{x\super{t} - \xopt} + \eta^2 \norm{\wg\super{t}}_2^2 \\
          & = \norm{x\super{t} - \xopt}_2^2 - 2 \eta \inp{g\super{t}}{x\super{t} - \xopt} - 2 \eta \inp{\wg\super{t} - g\super{t}}{x\super{t} - \xopt} + \eta^2 \norm{\wg\super{t}}_2^2 \\
          & \leq \norm{x\super{t} - \xopt}_2^2 - 2 \eta (f(x\super{t}) - f(\xopt)) - 2 \eta \inp{\wg\super{t} - g\super{t}}{x\super{t} - \xopt} + \eta^2 \norm{\wg\super{t}}_2^2
    \end{align*}
  where the last line is by the definition of a subgradient. We now take the expectation of the above formula. Using the law of total expectation, we have
    $\ex*{\inp{\wg\super{t} - g\super{t}}{x\super{t} - \xopt}}
    = \ex[\big]{\inp{\ex*{\wg\super{t} | x\super{t}} - g\super{t}}{x\super{t} - \xopt}}$
  and by H\"older's inequality
    $\abs*{\inp{\ex*{\wg\super{t} | x\super{t}} - g\super{t}}{x\super{t} - \xopt}}
    \leq \norm{\ex*{\wg\super{t} | x\super{t}} - g\super{t}}_1 \cdot \norm{x\super{t} - \xopt}_{\infty}
    \leq \eps L_{\infty}$.
  Consequently,
    \[\ex*{\norm{x\super{t+1} - \xopt}_2^2} - \ex*{\norm{x\super{t} - \xopt}_2^2} \leq  - 2 \eta \ex*{f(x\super{t}) - f(\xopt)} + 2 \eta \eps L_{\infty} + \eta^2 B^2\]
  from which we obtain a bound for $\ex{f(x\super{t})}$. Finally, we upper bound the expected value of the function at the average point $\bar{x}$ as
    \begin{align*}
     \ex*{f(\bar{x})}
        & \leq \frac{1}{N} \sum_{t=0}^{N-1} \ex*{f(x\super{t})} \quad \mbox{ by convexity}\\
        & \leq f(\xopt) + \frac{1}{N} \sum_{t=0}^{N-1} \frac{1}{2\eta} \left(\ex*{\norm{x\super{t} - \xopt}_2^2} - \ex*{\norm{x\super{t+1} - \xopt}_2^2}\right) + \frac{\eta}{2} B^2 + \eps L_{\infty} \\
        & = f(\xopt) + \frac{1}{2\eta N} \left(\ex*{\norm{x\super{0} - \xopt}_2^2} - \ex*{\norm{x\super{N} - \xopt}_2^2}\right) + \frac{\eta}{2} B^2  + \eps L_{\infty}\\
        & \leq f(\xopt) + \frac{L_2^2}{2\eta N} + \frac{\eta}{2} B^2 + \eps L_{\infty}
   \end{align*}
  where we have used the telescoping property of the sum in the third line.
\end{proof}

In the rest of the paper, $f$ denotes the \lv\ extension and $g$ denotes the \lv\ subgradient. Our main result of this section (Algorithm \ref{Algo:submDesc}) consists in constructing the sequence of noisy subgradient estimates needed in the above proposition. To trade off the cost of computing the subgradient exactly and decreasing the variance, we rely on two procedures that provide different guarantees on the estimates they return. In this section, we do not explain how to implement these two procedures. Instead, we describe in Assumptions \ref{Assp:one} and \ref{Assp:two} the main properties they must satisfy.

Our first assumption is the existence of a procedure $\samplg$ that can produce a batch of $T$ estimates of the \lv\ subgradient $g(x)$ at any point $x \in [0,1]^n$. This is intended to be a simple but expensive procedure, which can be used only sparingly. Indeed, it will need time $\bo{(n+T) \cdot \eo}$ or $\bo{(\sqrt{nT} + \sqrt{n}/\eps) \cdot \eo}$ to be implemented in the classical or quantum settings respectively (Propositions \ref{Prop:samplgCl} and \ref{Prop:samplgQu}).

\begin{assumption}[Gradient Sampling]
  \label{Assp:one}
  There is a procedure $\samplg(x,T,\eps)$ that, given $x \in [0,1]^n$, an integer $T$ and a real $\eps > 0$, outputs $T$ vectors $\wg^1, \dots, \wg^T$ such that, for all $j$, (1) $\wg^j$ is $1$-sparse, (2) $\big\lVert\ex{\wg^j | \wg^1, \dots, \wg^{j-1},x} - \gl(x)\big\rVert_1 \leq \eps$ and (3) $\norm{\wg^j}_2 \leq 4$. Moreover, the time complexity of this procedure is a function $\costg(T,\eps)$ of $T$ and $\eps$.
\end{assumption}

Our second assumption is the existence of a more subtle procedure $\sampld$ that can estimate the difference $g(y) - g(x)$ between the \lv\ subgradients at two points $x$ and $y$. This procedure will rely on intrinsic properties of submodular functions and require maintaining a particular data structure (Section \ref{Sec:datastruc}). On the other hand, when the difference $e = y - x$ is $k$-sparse, it will need time only $\so{k \cdot \eo}$ or $\so{\sqrt{k}/\eps \cdot \eo}$ to be implemented in the classical or quantum settings respectively (Propositions \ref{Prop:sampldCl} and \ref{Prop:sampldQu}).

\begin{assumption}[Gradient Difference Sampling]
  \label{Assp:two}
  There is a procedure $\sampld(x, e,\eps)$ that, given $x \in [0,1]^n$, a $k$-sparse vector $e$ such that $x + e \in [0,1]^n$ and $e \geq 0$ or $e \leq 0$, and a real $\eps > 0$, outputs a vector $\wdi$ such that, (1) $\wdi$ is $1$-sparse, (2) $\norm{\ex{\wdi | x, e} - (\gl(x + e) - \gl(x))}_1 \leq \eps$ and (3) $\norm{\wdi}_2 \leq 7$. Moreover, the time complexity of this procedure is a function $\costd(k,\eps)$ of $k$ and $\eps$.
\end{assumption}

We combine the two procedures to construct the sequence $(\wg\super{t})_t$ of subgradient estimates (Algorithm \ref{Algo:submDesc}). The construction depends on a ``loop parameter'' $T$ that balances the cost between using $\samplg$ and $\sampld$. Every $T$ steps, when $t = 0 \bmod T$, the procedure $\samplg(x\super{t},T,\eps)$ returns $T$ estimates $\wg\super{t,0}, \dots, \wg\super{t,T-1}$ of the \lv\ subgradient at the current point $x\super{t}$. Each value $\wg\super{t,\tau}$ is combined at time $t + \tau$, where $0 \leq \tau \leq T-1$, with an estimate $\wdi\super{t+\tau}$ of the subgradient difference $g(x\super{t + \tau}) - g(x\super{t})$. The sum $\wg\super{t + \tau} = \wg\super{t,\tau} + \wdi\super{t+\tau}$ is our estimate of $g(x\super{t + \tau})$. The sparsity of $x\super{t+\tau} - x\super{t}$ will increase linearly in $\tau$, which justifies reusing $\samplg$ every $T$ steps to restore it to a small value. Notice that, according to Assumption \ref{Assp:two}, the procedure $\sampld$ can estimate the subgradient difference $d = g(y) - g(x)$ only if $e = y - x$ is either non-negative or non-positive. Thus, in step 2.(c) of the algorithm, we split $e = e\ps + e\ms$ into its positive and negative entries and we estimate $d\ps = g(x + e\ps) - g(x)$ and $d\ms = g(x + e\ps + e\ms) - g(x + e\ps)$ separately. In the next theorem, we show that $(\wg\super{t})_t$ is indeed a sequence of noisy subgradient oracles for $\left(f, (x\super{t})_t\right)$.

\begin{algorithm}[htbp]
\caption{Subgradient descent algorithm for the \lv\ extension $f$.}
\label{Algo:submDesc}
\normalsize
\textbf{Input:} two integers $0 < T < N$, two reals $\eps_0, \eps_1 > 0$. \\
\textbf{Output:} point $\bar{x} \in [0,1]^n$. \vspace{0.1cm}

\begin{algorithmic}[1]
  \State Set $x\super{0} = 0^n \in [0,1]^n$.
  \For{$t = 0, \dots, N$}
    \State Set $\tau = (t \bmod T)$.
    \LineComment{\emph{Computation of the subgradient estimate $\wg\super{t}$:}}
    \State If $\tau = 0$: sample $\wg\super{t,0}, \dots, \wg\super{t,T-1}$ using $\samplg(x\super{t},T,\eps_0)$. Set $\wg\super{t} = \wg\super{t,0}$.
    \State \algparbox{If $\tau \neq 0$: sample $\wdi\super{t}\ps$ using $\sampld\bigl(x\super{t-\tau}, e\ps\super{t-1}, \eps_1\bigr)$ and $\wdi\super{t}\ms$ using $\sampld\bigl(x\super{t-\tau} \allowbreak + \allowbreak e\ps\super{t-1}, \allowbreak e\ms\super{t-1}, \eps_1\bigr)$. Set $\wg\super{t} = \wg\super{t-\tau,\tau} + \wdi\super{t}\ps + \wdi\super{t}\ms$.}
    \LineComment{\emph{Update of the position to $x\super{t+1}$:}}
    \State \algparbox{Compute $x\super{t+1} = \argmin_{x \in [0,1]^n} \|x - (x\super{t} - \eta \wg\super{t})\|_2$, that is $x\super{t+1} = x\super{t} + u\super{t}$ where
        \begin{align*}
        u\super{t}_i =
          \left\{
            \begin{array}{ll}
              - x\super{t}_i & \mbox{if $\eta \wg\super{t}_i > x\super{t}_i$} \\
              1 - x\super{t}_i & \mbox{if $\eta \wg\super{t}_i < - (1 - x\super{t}_i)$} \\
               - \eta \wg\super{t}_i & \mbox{otherwise}
            \end{array}
          \right.
        \end{align*}
        for each $i \in [n]$, and $\eta = \sqrt{\frac{n}{18^2 N}}$.}
    \LineComment{\emph{Update of the difference to $e\super{t} = x\super{t+1} - x\super{t-\tau}$:}}
    \State If $\tau = 0$, set $e\super{t} = u\super{t}$.
    \State If $\tau \neq 0$, set $e\super{t} = e\super{t-1} + u\super{t}$.
  \EndFor
  \State Output $\bar{x} = \frac{1}{N} \sum_{t=0}^{N-1} x\super{t}$.
\end{algorithmic}
\end{algorithm}

\begin{theorem}
  \label{Thm:gradSeq}
  The sequences $(\wg\super{t})_t$ and $(x\super{t})_t$ in Algorithm \ref{Algo:submDesc} satisfy $\norm{\ex*{\wg\super{t} | x\super{t}} - \gl(x\super{t})}_1 \allowbreak \leq \eps_0 + 2\eps_1$, $\norm{\wg\super{t}}_2 \leq 18$ and $x\super{t+1} = \argmin_{x \in [0,1]^n} \|x - (x\super{t} - \eta \wg\super{t})\|_2$.
\end{theorem}

\begin{proof}
  Fix $t$ and $\tau = (t \bmod T)$. According to lines 4 and 5 of the algorithm, we have
    \[
      \left\{
          \begin{array}{ll}
              \wg\super{t} = \wg\super{t,0} & \mbox{if } \tau = 0 \\
              \wg\super{t} = \wg\super{t-\tau,\tau} + \wdi\super{t}\ps + \wdi\super{t}\ms & \mbox{otherwise.}
          \end{array}
      \right.
    \]
We first study the expectation of the term $\wg\super{t-\tau,\tau}$, which is generated by the $\samplg$ procedure. Using the law of total expectation, it satisfies
    \begin{align*}
      \ex{\wg\super{t-\tau,\tau} | x\super{t}}
        & = \ex*{\ex{\wg\super{t-\tau,\tau} | (\wg\super{t-\tau,k})_{k < \tau}, x\super{t-\tau}, x\super{t} }| x\super{t}} \\
        & = \ex*{\ex{\wg\super{t-\tau,\tau} | (\wg\super{t-\tau,k})_{k < \tau}, x\super{t-\tau}} | x\super{t}}
    \end{align*}
  since $x\super{t}$ does not convey any information about the output of $\samplg(x\super{t-\tau},T)$ when $(\wg\super{t-\tau,k})_{k < \tau}$ and $x\super{t-\tau}$ are known. Consequently,
    \begin{align*}
      \norm*{\ex{\wg\super{t-\tau,\tau} - \gl(x\super{t-\tau})| x\super{t}}}_1
        & \leq \ex*{\norm*{\ex{\wg\super{t-\tau,\tau} | (\wg\super{t-\tau,k})_{k < \tau}, x\super{t-\tau}} - \gl(x\super{t-\tau})}_1 | x\super{t}} \\
        & \leq \eps_0
    \end{align*}
  using the triangle inequality and Assumption \ref{Assp:one}.

  We now study the expectation of the term $\wdi\super{t}\ps + \wdi\super{t}\ms $ generated by the $\sampld$ procedure when $\tau \neq 0$. We have
    \begin{align*}
      \ex{\wdi\super{t}\ps + \wdi\super{t}\ms | x\super{t}}
        & = \ex*{\ex{\wdi\super{t}\ps|x\super{t-\tau},e\ps\super{t-1},x\super{t}} + \ex{\wdi\super{t}\ms|x\super{t-\tau} + e\ps\super{t-1},e\ms\super{t-1},x\super{t}} | x\super{t}} \\
        & = \ex*{\ex{\wdi\super{t}\ps|x\super{t-\tau},e\ps\super{t-1}} + \ex{\wdi\super{t}\ms|x\super{t-\tau} + e\ps\super{t-1},e\ms\super{t-1}} | x\super{t}}
    \end{align*}
  where the first line is by the law of total expectation, and the second line is by independence between random variables. Moreover, according to Assumption \ref{Assp:two}, $\norm{\ex{\wdi\super{t}\ps|x\super{t-\tau},e\ps\super{t-1}} - (\gl(x\super{t-\tau} + e\ps\super{t-1}) - \gl(x\super{t-\tau}))}_1 \leq \eps_1$ and $\norm{\ex{\wdi\super{t}\ms|x\super{t-\tau} + e\ps\super{t-1},e\ms\super{t-1}} - (\gl(x\super{t}) - \gl(x\super{t-\tau} + e\ps\super{t-1}))}_1 \leq \eps_1$ (where we used that $x\super{t} = x\super{t-\tau} + e\ps\super{t-1} + e\ms\super{t-1}$). Thus, by the triangle inequality,
    \[\norm{\ex{\wdi\super{t}\ps + \wdi\super{t}\ms - (\gl(x\super{t}) - \gl(x\super{t-\tau}))| x\super{t}}}_1 \leq 2\eps_1.\]
  This concludes the proof of the first part of the theorem since
    $\norm{\ex{\wg\super{t} | x\super{t}}- \gl(x\super{t})}_1
    = \norm{\ex{\wg\super{t,0} - \allowbreak \gl(x\super{t})| x\super{t}}}_1
    \leq \eps_0$
    when $\tau = 0$, and
    $\norm{\ex{\wg\super{t} | x\super{t}}- \gl(x\super{t})}_1
      \leq \norm{\ex{\wg\super{t-\tau,\tau} - \allowbreak \gl(x\super{t-\tau})| x\super{t}}}_1
          + \norm{\ex{\wdi\super{t}\ps + \wdi\super{t}\ms - (\gl(x\super{t}) - \gl(x\super{t-\tau}))| x\super{t}}}_1
      \leq \eps_0 + 2 \eps_1$
    when $\tau \neq 0$.
  The second part of the theorem is a direct application of the triangle inequality using that $\norm{\wg\super{t-\tau,\tau}}_2 \leq 4$ and $\norm{\wdi\super{t}\ps}_2, \allowbreak \norm{\wdi\super{t}\ms}_2 \leq 7$ (Assumptions \ref{Assp:one} and \ref{Assp:two}). The last part of the theorem is line 6 of the algorithm.
\end{proof}

The above result shows that Algorithm \ref{Algo:submDesc} is a (noisy) subgradient descent for the \lv\ extension. Consequently, the result of Proposition \ref{Prop:gradDescent} can be applied to the output $\bar{x}$ of the algorithm. Since we aim for a subquadratic running time in $n$, we must update the vectors $x\super{t}$, $\wg\super{t}$, $u\super{t}$ and $e\super{t}$ in time less than their dimension. Here, we do not discuss the data structure used for this purpose (see Section \ref{Sec:datastruc}). Nevertheless, we recall that the outputs of $\samplg$ and $\sampld$ are $1$-sparse, thus most of the coordinates do not change between two consecutive steps.

\begin{fact}
  \label{Fac:sparse}
  At step $t$ of the algorithm: $\wg\super{t}$ and $u\super{t}$ are $3$-sparse, $e\ps\super{t}$ and $e\ms\super{t}$ are $3(\tau+1)$-sparse, if $\tau \neq 0$ then $e\ps\super{t-1}$ and $e\ps\super{t}$ (resp. $e\ms\super{t-1}$ and $e\ms\super{t}$) can differ only at positions where $u\super{t}$ is non-zero.
\end{fact}

\begin{corollary}
  \label{Cor:submDesc}
  The output $\bar{x}$ of Algorithm \ref{Algo:submDesc} satisfies $\ex{f(\bar{x})} \leq \min_x \fl(x) + 18 \sqrt{n/N} + \eps_0 + 2\eps_1$. The total run-time of steps 2.(b) and 2.(c) is $\bo*{\frac{N}{T} \Bigl(\costg(T,\eps_0) + \sum_{\tau = 1}^{T} \costd(3\tau,\eps_1)\Bigr)}$.
\end{corollary}

\begin{proof}
  According to Theorem \ref{Thm:gradSeq}, $(\wg\super{t})_t$ is a sequence of $\eps$-noisy subgradient oracles for the \lv\ extension $\fl$, where $\eps = \eps_0+2\eps_1$ and $(x\super{t})_t$ obeys the subgradient descent update rule $x\super{t+1} = \argmin_{x \in [0,1]^n} \norm{x - (x\super{t} - \eta \wg\super{t})}_2$. Moreover, $\norm{x - \xopt}_2 \leq \sqrt{n}$, $\norm{x - \xopt}_{\infty} \leq 1$ for all $x \in [0,1]^n$, and $\norm{\wg\super{t}}_2 \leq 18$. Consequently, we obtain from Proposition \ref{Prop:gradDescent} that $\ex*{\fl(\bar{x})} \leq \fl(\xopt) + 18 \sqrt{n/N} + \eps_0 + 2\eps_1$, where we used the step size parameter $\eta = \sqrt{\frac{n}{18^2 N}}$. The time complexity of steps 2.(b) and 2.(c) is a direct consequence of Assumptions \ref{Assp:one} and \ref{Assp:two} and Fact \ref{Fac:sparse}.
\end{proof}

%% file: Datastructure.tex
We describe two data structures needed to implement $\samplg$ and $\sampld$ respectively, and we establish their main properties. The first data structure $\ds(x)$ contains standard information about the input point $x$ of $\samplg(x,T,\eps)$. It is defined as follows.

\begin{definition}
  \label{Def:dsSimple}
  Given $x \in \R^n$ and the permutation $P$ consistent with $x$, we define $\ds(x) = (L_x,A_x,T_x)$ to be the data structure made of the following elements: a doubly linked list $L_x$ storing $P$, an array $A_x$ storing at position $i \in [n]$ the value $x_i$ with a pointer to the corresponding entry in $P$, and a self-balancing binary search tree $T_x$ (e.g. red-black tree \cite{GS78c}) with a node for each $i \in [n]$ keyed by the value $x_i$ and containing the size of its subtree.
\end{definition}

The second data structure $\ds(x,y,\mathcal{I})$ is based on a property of submodular functions established in \cite{CLSW17c} that requires the following definition of a $k$-cover.

\begin{definition}
  \label{Def:monot}
  Consider $x, y \in [0,1]^n$ and let $P$ and $Q$ be the permutations consistent with $x$ and $y$ respectively. We say that a partition $\mathcal{I} = \{I_1, \dots, I_k\}$ of $[n]$ is a \emph{$k$-cover} of $(x,y)$ if, for each $s \in [k]$, the preimage of $I_s$ under both $P$ and $Q$ is a set of consecutive numbers, and $x_i = y_i$ for all $i \in I_s$ if $\abs{I_s} > 1$.
\end{definition}

\begin{figure}[htbp]
    \centering
   \begin{tabular}{c@{\hspace{1cm}}c}
     \scalebox{0.8}{
     \addtolength{\tabcolsep}{-3pt}
      \begin{tabular}{c | cccccccccc}
         & 1 & 2 & 3 & 4 & 5 & 6 & 7 & 8 & 9 & 10 \\ \hline
       $x$ & 0.85 & 0.58 & \textbf{0.42} & 0.53 & 0.60 & 0.78 & \textbf{0.12} & 0.27 & 0.92 & 0.31 \\
       $y$ & 0.85 & 0.58 & \textbf{0.65} & 0.53 & 0.60 & 0.78 & \textbf{0.90} & 0.27 & 0.92 & 0.31 \\ \hline
       $P$ & 9 & 1 & 6 & 5 & 2 & 4 & \circled{3} & 10 & 8 & \circled{7} \\
       $Q$ & 9 & \circled{7} & 1 & 6 & \circled{3} & 5 & 2 & 4 & 10 & 8 \\ \hline
     \end{tabular}}
   &
    \begin{tikzpicture}[baseline=(current bounding box.center),every node/.style={circle split,draw,scale=0.65,inner sep=1.5pt},scale=0.75]
    \node{$0.60$\nodepart{lower}\mbox{\Large$I_4$}}
      child{
        node{$0.31$\nodepart{lower}\mbox{\Large$I_2$}}
        child{node{$0.12$\nodepart{lower}\mbox{\Large$I_3$}}}
        child{node{$0.42$\nodepart{lower}\mbox{\Large$I_1$}}}
      }
      child{
        node{$0.85$\nodepart{lower}\mbox{\Large$I_5$}}
        child[missing]{}
        child{node{$0.92$\nodepart{lower}\mbox{\Large$I_6$}}}
      };
    \end{tikzpicture}
  \end{tabular}

   \vspace{0.4cm} $I_1 = \{3\}$, $I_2 = \{10,8\}$, $I_3 = \{7\}$, $I_4 = \{5,2,4\}$, $I_5 = \{1,6\}$, $I_6 = \{9\}$.
   \label{Fig:cover}
   \vspace*{13pt}
   \caption{An illustration of a $6$-cover $\mathcal{I} = \{I_1,\dots,I_6\}$ for some $x, y \in [0,1]^{10}$ and their corresponding permutations $P, Q$. The circled numbers correspond to the positions where $x$ and $y$ differ (these values must belong to singletons in the cover). The binary tree corresponds to $T_x^\mathcal{I}$ in $\ds(x,y,\mathcal{I})$.}
\end{figure}
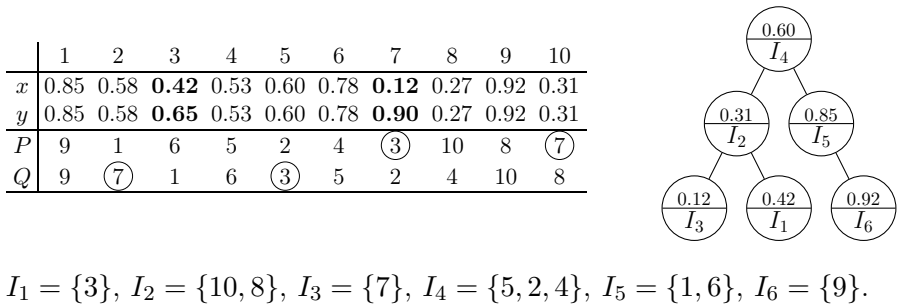

An example of a $6$-cover is given in Figure \ref{Fig:cover}. Observe that there always exists a cover of size at most $3k+1$ if the difference $e = x - y$ is $k$-sparse. Our implementations of $\sampld(x,e,\eps)$ will require to store a cover of $(x,x+e)$ approaching that size. Before explaining the reasons why a cover is useful, we describe the data structure $\ds(x,y,\mathcal{I})$.

\begin{definition}
  Given $x, y \in \R^n$ and a $k$-cover $\mathcal{I} = \{I_1, \dots, I_k\}$ of $(x,y)$, we define $\ds(x,y,\mathcal{I}) = \left(\ds(x),\ds(y),A_x^\mathcal{I},A_y^\mathcal{I},T_x^\mathcal{I},T_y^\mathcal{I}\right)$ to be the data structure made of the following elements: $\ds(x)$ and $\ds(y)$ (described in Definition \ref{Def:dsSimple}), two dynamic arrays $A_x^\mathcal{I}$ and $A_y^\mathcal{I}$ of size $k$ storing at position $s \in [k]$ the pairs $(\argmax_{i \in I_s} x_i, \allowbreak \argmin_{i \in I_s} x_i)$ and $(\argmax_{i \in I_s} y_i,\argmin_{i \in I_s} y_i)$ respectively, two self-balancing binary search trees $T_x^\mathcal{I}$ and $T_y^\mathcal{I}$ with a node for each $s \in [k]$ keyed by the value of $\max_{i \in I_s} x_i$ and $\max_{i \in I_s} y_i$ respectively.
\end{definition}

The next lemma is the crucial property established in \cite{CLSW17c} about $k$-covers. It shows that the coordinates $\gl(y)_i - \gl(x)_i$ of the subgradient difference have constant sign over any set $I_s$ of the cover when $y \geq x$ or $y \leq x$. In particular, the $\ell_1$-norm $\norm{\gl(y)_{I_s} - \gl(x)_{I_s}}_1$ can be deduced from the value of $\sum_{i \in I_s} \gl(y)_i - \gl(x)_i$ (we recall that $\gl(y)_{I_s} - \gl(x)_{I_s}$ is the subvector of $\gl(y) - \gl(x)$ made of the values at positions $i \in I_s$).

\begin{lemma}[\cite{CLSW17c}]
  \label{Lem:normCover}
  Consider $x, y \in [0,1]^n$ such that $y \geq x$ or $y \leq x$, and let $\{I_1, \dots, I_k\}$ be a $k$-cover of $(x,y)$. Then, for each $s \in [k]$, the coordinates $\gl(y)_i - \gl(x)_i$ have the same sign for all $i \in I_s$. In particular, $\abs{\sum_{i \in I_s} \gl(y)_i - \gl(x)_i} = \norm{\gl(y)_{I_s} - \gl(x)_{I_s}}_1$.
\end{lemma}

\begin{proof}
  Let $P$ and $Q$ denote the permutations consistent with $x$ and $y$ respectively. Consider $s \in [k]$ such that $\abs{I_s} > 1$ (the result is trivial when $\abs{I_s} = 1$). By definition of a $k$-cover, there exist three integers $a_s$, $a'_s$, $\ell_s$ such that $I_s = \{P_{a_s}, P_{a_s + 1}, \dots, P_{a_s + \ell_s}\} = \{Q_{a'_s}, Q_{a'_s + 1}, \dots, Q_{a'_s + \ell_s}\}$. Assume that $y \geq x$ (the case $y \leq x$ is symmetric). Since $x_i = y_i$ for all $i \in I_s$, we must have $P[a_s-1] \subseteq Q[a'_s-1]$. Thus, by the diminishing returns property of submodular functions, $F(P[a_s + \ell]) - F(P[a_s + \ell - 1]) \geq F(Q[a'_s + \ell]) - F(Q[a'_s + \ell - 1])$ for all $0 \leq \ell \leq \ell_s$. We conclude that $\gl(y)_i - \gl(x)_i \leq 0$ for all $i \in I_s$.
\end{proof}

Note that the condition $y \geq x$ or $y \leq x$ is crucial in the above result, which is why we impose $e \geq 0$ or $e \leq 0$ in Assumption \ref{Assp:two}. We now describe three useful operations that can be handled in logarithmic time using $\ds(x,y,\mathcal{I})$ and the above lemma. The first two operations originate from the work of \cite{CLSW17c} (the proofs are given for completeness), whereas the third one is new to this work (in \cite{CLSW17c}, the authors recompute the cover entirely at each update).

\begin{proposition}
  \label{Prop:dataOpe}
  Consider $x, y \in \R^n$ such that $x \geq y$ or $x \leq y$, and let $\mathcal{I} = \{I_1, \dots, I_k\}$ be a $k$-cover of $(x,y)$. Then, using $\ds(x,y,\mathcal{I})$, the following operations can be handled in time $\bo{\log(n) + \eo}$, $\bo{\log(n) \cdot \eo}$ and $\bo{\log n}$ respectively:
    \begin{itemize}
      \item \emph{(Subnorm)} Given $s \in [k]$, output $\norm{\gl(y)_{I_s} - \gl(x)_{I_s}}_1$.
      \item \emph{(Subsampling)} Given $s \in [k]$, sample $i \sim \dis{\gl(y)_{I_s} - \gl(x)_{I_s}}$.
      \item \emph{(Update)} Given a $1$-sparse vector $e \in \R^n$, update the data structure to $\ds(x+e,y,\mathcal{I}')$ where $\mathcal{I}'$ is a cover of $(x+e,y)$ of size at most $k + 3$.
    \end{itemize}
\end{proposition}

\begin{proof}
  Let $P$ be the permutation consistent with $x$. Observe that the rank $P^{-1}_i$ of any $x_i$ can be computed in $\log(n)$ time using $T_x$ (since each node in the tree contains the size of its subtree).

  \emph{(Subnorm)} By definition of a $k$-cover, there exist $a_s \leq b_s$ such that $I_s = \{P_{a_s}, P_{a_s + 1}, \dots, \allowbreak P_{b_s}\}$. Thus, $\sum_{i \in I_s} \gl(x)_i = \sum_{i=a_s}^{b_s} F(P[i]) - F(P[i-1]) = F(P[b_s]) - F(P[a_s-1])$. Since $a_s$ and $b_s$ can be obtained in time $\bo{\log n}$ using $\ds(x,y,\mathcal{I})$, this sum can be computed in time $\bo{\log(n) + \eo}$, and similarly for $\sum_{i \in I_s} \gl(y)_i$. According to Lemma \ref{Lem:normCover}, the difference $\abs{\sum_{i \in I_s} \gl(y)_i - \gl(x)_i}$ is equal to the $\ell_1$-norm of $\gl(y)_{I_s} - \gl(x)_{I_s}$.

  \emph{(Subsampling)} We find the highest node $i_h \in I_s$ in the tree $T_x$, and we compute its rank $c_s = P^{-1}_{i_h}$ in time $\bo{\log n}$. It partitions $I_s$ into three sets $A = \{P_{a_s}, P_{a_s + 1}, \dots, P_{c_s - 1}\}$, $B = \{P_{c_s}\}$ and $C = \{P_{c_s + 1}, \dots, P_{b_s}\}$. We compute the $\ell_1$-norm of $\gl(y) - \gl(x)$ restricted to each of these sets in time $\bo{\eo}$, and we sample $A$, $B$ or $C$ with probability $\frac{\norm{\gl(y)_{A} - \gl(x)_{A}}_1}{\norm{\gl(y)_{I_s} - \gl(x)_{I_s}}_1}$, $\frac{\norm{\gl(y)_{B} - \gl(x)_{B}}_1}{\norm{\gl(y)_{I_s} - \gl(x)_{I_s}}_1}$ and $\frac{\norm{\gl(y)_{C} - \gl(x)_{C}}_1}{\norm{\gl(y)_{I_s} - \gl(x)_{I_s}}_1}$ respectively. If the set we obtain is a singleton, we terminate and output the value it contains, otherwise we continue recursively to sample in the corresponding subtree of $T_x$.

  \emph{(Update)} We explain how to update the $k$-cover (the other parts of the data structure being standard to update). Let $i$ be the position where $e_i \neq 0$, and denote $r = P^{-1}_i$ the rank of $i$ in $P$ before the update. First, split the set $I_s = \{P_{a_s}, \dots, P_{r-1}, i, P_{r+1}, \dots, P_{b_s}\}$ containing $i$ into three parts $\{P_{a_s}, P_{a_s + 1}, P_{r-1}\}$, $\{i\}$ and $\{P_{r+1}, \dots, P_{b_s}\}$. Then, identify the rank $c$ such that $x_{P_c} > x_i + e_i > x_{P_{c+1}}$ and split the set containing $P_c$ into two parts. These operations can be done in $\bo{\log n}$ time. The size of the cover is increased by at most $3$.
\end{proof}

Finally, observe that we have to maintain two instances of $\ds(x,y,\mathcal{I})$ in Algorithm \ref{Algo:submDesc}, one corresponding to the pair $(x\super{t-\tau}, \allowbreak x\super{t-\tau} + e\ps\super{t-1})$ (needed for $\wdi\super{t}\ps$), and the other one corresponding to $(x\super{t-\tau} + e\ps\super{t-1}, x\super{t})$ (needed for $\wdi\super{t}\ms$). The total update cost is negligible because of Fact \ref{Fac:sparse}.

\begin{corollary}
  \label{Cor:dataMaintain}
  One can maintain throughout Algorithm \ref{Algo:submDesc} two data structures $\ds(x\super{t-\tau}, \allowbreak x\super{t-\tau} + e\ps\super{t-1}, \mathcal{I})$ and $\ds(x\super{t-\tau} + e\ps\super{t-1}, x\super{t}, \mathcal{I}')$, where $\mathcal{I}$ and $\mathcal{I}'$ are covers of size at most $9 \tau$ and $27 \tau$ respectively. The update time, at each step of the algorithm, is $\bo{\log n}$.
\end{corollary}

\begin{proof}
  This is a direct consequence of Fact \ref{Fac:sparse} and Proposition \ref{Prop:dataOpe}. The update from $e\ps\super{t-1}$ to $e\ps\super{t}$ (resp. $e\ms\super{t-1}$ to $e\ms\super{t}$) is $3$-sparse, thus each step increases the size of the cover associated with $(x\super{t-\tau}, \allowbreak x\super{t-\tau} + e\ps\super{t-1})$ by $9$, and the size of the cover associated with $(x\super{t-\tau} + e\ps\super{t-1}, x\super{t}) = (x\super{t-\tau} + e\ps\super{t-1}, x\super{t-\tau} + e\ps\super{t-1} + e\ms\super{t-1})$ by $27$. When $\tau = 0$, the sizes are reset to at most $9$ and $27$ respectively.
\end{proof}

%% file: ClassicalApproximation.tex
We conclude the part on classical approximate submodular minimization by describing the two procedures $\cgs$ and $\cgds$ for $\samplg$ and $\sampld$ respectively that lead to an $\so{n^{3/2}/\eps^2 \cdot \eo}$ algorithm. Both are based on the following unbiased estimator $X_u$ of any vector $u \in \R^n$.

\begin{fact}
  \label{Fact:estUnbias}
  Given a non-zero vector $u \in \R^n$, consider the vector-valued random variable $X_u$ that equals $\norm{u}_1 \sgn(u_i) \cdot \vec{1}_i$ with probability $\frac{\abs{u_i}}{\norm{u}_1}$. Then, $\ex{X_u} = u$ and $\norm{X_u}_2 = \norm{u}_1$.
\end{fact}

In the case of gradient sampling (Assumption \ref{Assp:one}), we construct $T$ samples from $X_{g(x)}$ by using the sampling algorithm of Lemma \ref{Lem:samplClass}. In the case of gradient difference sampling (Assumption \ref{Assp:two}), we construct one sample from $X_{g(x+e) - g(x)}$ by using the subnorm and subsampling operations of Proposition \ref{Prop:dataOpe}. These two procedures are described in Algorithms \ref{Algo:CGS} and \ref{Algo:CGDS} respectively.

\begin{algorithm}[htbp]
\caption{Classical Gradient Sampling ($\cgs$).}
\label{Algo:CGS}
\normalsize
\textbf{Input:} $x \in [0,1]^n$ stored in $\ds(x)$, an integer $T$. \\
\textbf{Output:} a sequence of estimates $\wg^1, \dots, \wg^T$ of $g(x)$.

\begin{algorithmic}[1]
  \State Compute $\norm{g(x)}_1$ and sample $i_1, \dots, i_T \sim \dis{g(x)}$ using Lemma \ref{Lem:samplClass}.
  \State For each $j \in [T]$, compute $g(x)_{i_j}$ and output $\wg^j = \norm{g(x)}_1 \sgn(g(x)_{i_j}) \cdot \vec{1}_{i_j}$.
\end{algorithmic}
\end{algorithm}

\begin{proposition}
  \label{Prop:samplgCl}
  The classical procedure $\cgs(x,T)$ of Algorithm \ref{Algo:CGS} satisfies the conditions given on $\samplg(x,T,0)$ in Assumption \ref{Assp:one}, with time complexity $\costg(T,0) = \bo{(n + T) \cdot \eo}$.
\end{proposition}

\begin{proof}
  By Fact \ref{Fact:estUnbias}, we have $\ex{\wg^j | \wg^1, \dots, \wg^{j-1}, x} = \ex{\wg^j | x} = g(x)$ and $\norm{\wg^j}_2 = \norm{g(x)}_1 \leq 3$. Line 1 takes time $\bo{n \cdot \eo + T}$, and line 2 takes time $\bo{T \cdot \eo}$.
\end{proof}

\begin{algorithm}[htbp]
\caption{Classical Gradient Difference Sampling ($\cgds$).}
\label{Algo:CGDS}
\normalsize
\textbf{Input:} $x, x+e \in [0,1]^n$ and a cover $\mathcal{I} = \{I_1, \dots, I_{k'}\}$ of $(x,x+e)$ stored in $\ds(x,x+e,\mathcal{I})$, where $e$ is $k$-sparse and $k' \leq 9k$. \\
\textbf{Output:} an estimate $\wdi$ of $g(x+e) - g(x)$. \vspace{0.1cm}

Let $u = (\norm*{g(x+e)_{I_s} - g(x)_{I_s}}_1)_{s \in [k']} \in \R^{k'}$. \vspace{0.1cm}

\begin{algorithmic}[1]
  \State Compute $\norm{u}_1$ and sample $s \sim \dis{u}$ using Lemma \ref{Lem:samplClass} and the subnorm operation of Proposition \ref{Prop:dataOpe}.
  \State Sample $i \sim \dis{g(x+e)_{I_s} - g(x)_{I_s}}$ using the subsampling operation of Proposition \ref{Prop:dataOpe}.
  \State Compute $g(x+e)_i - g(x)_i$ and output $\wdi = \norm{u}_1 \sgn(g(x+e)_i - g(x)_i) \cdot \vec{1}_i$.
\end{algorithmic}
\end{algorithm}

\begin{proposition}
  \label{Prop:sampldCl}
  The classical procedure $\cgds(x,e)$ of Algorithm \ref{Algo:CGDS} satisfies the conditions given on $\sampld(x,e,0)$ in Assumption \ref{Assp:two}, with time complexity $\costd(k,0) = \so{k \cdot \eo}$.
\end{proposition}

\begin{proof}
  The value $i$ is distributed according to $\frac{\norm{g(x+e)_{I_s} - g(x)_{I_s}}_1}{\norm{u}_1} \cdot \allowbreak \frac{\abs{\gl(x+e)_i - \gl(x)_i}}{\norm{\gl(x+e)_{I_s} - \gl(x)_{I_s}}_1} = \frac{\abs{\gl(x+e)_i - \gl(x)_i}}{\norm{\gl(x+e) - \gl(x)}_1}$ (since $\norm{u}_1 = \norm{g(x+e) - g(x)}_1$). This corresponds to $\dis{g(x+e) - g(x)}$. Consequently, using Fact \ref{Fact:estUnbias}, $\ex{\wdi | x, e} = g(x+e) - g(x)$ and $\norm{\wdi}_2 = \norm{g(x+e) - g(x)}_1 \leq 6$. Line 1 takes time $\bo{k (\log(n) + \eo)}$, line 2 takes time $\bo{\log(n) \cdot \eo}$ and line 3 takes time $\bo{\eo}$.
\end{proof}

Finally, we analyze the cost of using the two above procedures in the subgradient descent algorithm studied in Section \ref{Sec:frame}.

\begin{theorem}
  \label{Thm:clSub}
  There is a classical algorithm that, given a submodular function $F : 2^V \ra [-1,1]$ and $\eps > 0$, computes a set $\bar{S}$ such that $\ex{{F}(\bar{S})} \leq \min_{S \subseteq V} F(S) + \eps$ in time $\so{n^{3/2}/\eps^2 \cdot \eo}$.
\end{theorem}

\begin{proof}
  We instantiate Algorithm \ref{Algo:submDesc} with the procedures $\cgs$ and $\cgds$ of Algorithms \ref{Algo:CGS} and \ref{Algo:CGDS} respectively, and we choose the input parameters $T = \sqrt{n}$, $N = 18^2n/\eps^2$ and $\eps_0 = \eps_1 = 0$. The data structures needed for $\cgs$ and $\cgds$ can be updated in time $\bo{\log n}$ per step (Corollary \ref{Cor:dataMaintain}). According to Corollary \ref{Cor:submDesc}, we obtain an output $\bar{x}$ such that $\ex{f(\bar{x})} \leq \min_x f(x) + \eps$ in time $\so*{\frac{\sqrt{n}}{\eps^2} \left(n + \sum_{\tau = 1}^{\sqrt{n}} \tau\right) \cdot \eo} = \so{n^{3/2}/\eps^2 \cdot \eo}$.
  Finally, using Proposition \ref{Prop:lv}, we can convert $\bar{x}$ into a set $\bar{S} \subseteq V$ such that $\ex{{F}(\bar{S})} \leq \min_{S \subseteq V} F(S) + \eps$ in time $\bo{n \log n + n \cdot \eo}$.
\end{proof}

%% file: QuantumApproximation.tex
We conclude the part on quantum approximate submodular minimization by describing the two procedures $\qgs$ and $\qgds$ for $\samplg$ and $\sampld$ respectively that lead to an $\so{n^{5/4}/\eps^{5/2} \cdot \log(1/\eps) \cdot \eo}$ algorithm. Both are based on the following noisy estimator $Y_u$ of any vector $u \in \R^n$.

\begin{fact}
  \label{Fact:estBias}
  Given a non-zero vector $u \in \R^n$, a real $\Gamma > 0$ and a set $S \subseteq [n]$ such that $\Gamma \geq \norm{u_S}_1$, consider the vector-valued random variable $Y_u$ that equals $\Gamma \sgn(u_i) \cdot \vec{1}_i$ where $i \sim \dis{u}(\Gamma,S)$. Then, $\norm{\ex{Y_u} - u}_1 = \abs{\Gamma - \norm{u}_1}$ and $\norm{Y_u}_2 = \Gamma$.
\end{fact}

In the case of gradient sampling (Assumption \ref{Assp:one}), we construct $T$ samples from $Y_{g(x)}$ by using the sampling algorithm of Theorem \ref{Thm:KSampl}. In the case of gradient difference sampling (Assumption \ref{Assp:two}), we construct one sample from $Y_{g(x+e) - g(x)}$ by using Lemma \ref{Lem:samplOne} and the following evaluation oracle derived from Proposition \ref{Prop:dataOpe}.

\begin{proposition}
  \label{Prop:evalOracle}
  There is a quantum algorithm, represented as a unitary operator $\mathcal{O}$, that given $x,y \in [0,1]^n$ and a $k$-cover $\mathcal{I} = \{I_1, \dots, I_{k}\}$ of $(x,y)$ stored in $\ds(x,y,\mathcal{I})$, satisfies $\mathcal{O}(\ket{s} \ket{0}) = \ket[\big]{s} \ket[\big]{\norm{\gl(y)_{I_s} - \gl(x)_{I_s}}_1}$, for all $s \in [k]$. The time complexity of this algorithm is $\bo{\log(n) + \eo}$.
\end{proposition}

The procedures $\samplg$ and $\sampld$ are described in Algorithms \ref{Algo:QGS} and \ref{Algo:QGDS} respectively. There is a non-zero probability that the computation of the setup parameters is incorrect. In this case, we cannot guarantee that Assumptions \ref{Assp:one} and \ref{Assp:two} are satisfied. Fortunately, the dependence of the time complexity on the inverse failure probability is logarithmic. Thus, it will not impact the analysis significantly.

\begin{algorithm}[htbp]
\caption{Quantum Gradient Sampling ($\qgs$).}
\label{Algo:QGS}
\normalsize
\textbf{Input:} $x \in [0,1]^n$ stored in $\ds(x)$, an integer $T$, two reals $0 < \eps, \delta < 1$. \\
\textbf{Output:} a sequence of estimates $\wg^1, \dots, \wg^T$ of $g(x)$.

\begin{algorithmic}[1]
  \State Compute the setup parameters $(\Gamma,S,M)$ using Proposition \ref{Prop:setupShort} with input $g(x)$, $T$, $\eps/3$, $\delta$.
  \State Sample $i^1, \dots, i^T \sim \dis{g(x)}(\Gamma,S)$ using Theorem \ref{Thm:KSampl} with input $g(x)$, $T$, $(\Gamma,S,M)$.
  \State For each $j \in [T]$, compute $g(x)_{i_j}$ and output $\wg^j = \Gamma \sgn(g(x)_{i_j}) \cdot \vec{1}_{i_j}$.
\end{algorithmic}
\end{algorithm}

\begin{proposition}
  \label{Prop:samplgQu}
  The quantum procedure $\qgs(x,T,\eps,\delta)$ of Algorithm \ref{Algo:QGS} satisfies the conditions given on $\samplg(x,T,\eps)$ in Assumption \ref{Assp:one} with probability $1-\delta$, and time complexity $\costd(T,\eps) = \bo{(\sqrt{nT} + \sqrt{n}/\eps)\log(1/\delta) \cdot \eo}$.
\end{proposition}

\begin{proof}
  Let us denote $\gd$ the event that the setup parameters $(\Gamma,S,M)$ computed at line $1$ of the algorithm are valid (i.e. they satisfy the properties given in Proposition \ref{Prop:setupShort}). We have $\pb{\gd} \geq 1-\delta$. According to Theorem \ref{Thm:KSampl}, if $\gd$ holds, then $i^1, \dots, i^T$ are $T$ independent samples from $\dis{g(x)}(\Gamma,S)$. Consequently, using Fact \ref{Fact:estBias}, we have $\norm{\ex{\wg^j | x, (\Gamma, S, M), \gd} - g(x)}_1 = \abs{\Gamma - \norm{g(x)}_1} \leq (\eps/3) \norm{g(x)}_1 \leq \eps$ and $\norm{\wg^j}_2 \leq (1+\eps/3) \norm{g(x)}_1 \leq 4$. Moreover, since $\wg^1, \dots, \wg^T$ are independent \emph{conditioned on} $(\Gamma,S,M)$, by the law of total expectation $\norm{\ex{\wg^j | x, \wg^1, \dots, \wg^{j-1}, \gd} - g(x)}_1 = \norm[\big]{\ex[\big]{\ex{\wg^j | x, (\Gamma, S, M), \gd} | \wg^1, \dots, \wg^{j-1}, \gd} - g(x)}_1 \leq \eps$. Finally, line 1 takes time $\bo{(\sqrt{nT} + \sqrt{n}/\eps)\log(1/\delta) \cdot \eo}$, line 2 takes time $\bo{\sqrt{nT} \cdot \eo}$, and line 3 takes time $\bo{T \cdot \eo}$.
\end{proof}

\begin{algorithm}[htbp]
\caption{Quantum Gradient Difference Sampling ($\qgds$).}
\label{Algo:QGDS}
\normalsize
\textbf{Input:} $x, x+e \in [0,1]^n$ and a cover $\mathcal{I} = \{I_1, \dots, I_{k'}\}$ of $(x,x+e)$ stored in $\ds(x,x+e,\mathcal{I})$, where $e$ is $k$-sparse and $k' \leq 9k$, two reals $0 < \eps, \delta < 1$. \\
\textbf{Output:} an estimate $\wdi$ of $g(x+e) - g(x)$. \vspace{0.1cm}

Let $u = (\norm*{g(x+e)_{I_s} - g(x)_{I_s}}_1)_{s \in [k']} \in \R^{k'}$. \vspace{0.1cm}

\begin{algorithmic}[1]
  \State Compute $\norm{u}_{\infty}$ with success probability $1-\delta/2$, using D\"urr-H{\o}yer's algorithm \cite{DH96p} and Proposition \ref{Prop:evalOracle}. Denote the result by $M$.
  \State Compute an estimate $\Gamma$ of $\norm{u}_1$ with relative error $\eps/6$ and success probability $1-\delta/2$, using Lemma \ref{Lem:amplEst} and Proposition \ref{Prop:evalOracle}.
  \State Sample $s \sim \dis{u}$ using Lemma \ref{Lem:samplOne} on input $u$ and $M$.
  \State Sample $i \sim \dis{g(x+e)_{I_s} - g(x)_{I_s}}$ using the subsampling operation of Proposition \ref{Prop:dataOpe}.
  \State Compute $g(x+e)_i - g(x)_i$ and output $\wdi = \Gamma \sgn(g(x+e)_i - g(x)_i) \cdot \vec{1}_i$.
\end{algorithmic}
\end{algorithm}

\begin{proposition}
  \label{Prop:sampldQu}
  The quantum procedure $\qgds(x,e,\eps,\delta)$ of Algorithm \ref{Algo:QGDS} satisfies the conditions given on $\sampld(x,e,\eps)$ in Assumption \ref{Assp:two} with probability $1-\delta$, and time complexity $\costg(k,\eps) = \so{\sqrt{k}/\eps \cdot \log(1/\delta) \cdot \eo}$.
\end{proposition}

\begin{proof}
  Let us denote $\gd$ the event that $\Gamma$ and $M$ are valid, i.e. $\abs{\Gamma - \norm{g(x+e) - g(x)}_1} \leq (\eps/6) \norm{g(x+e) - g(x)}_1$ and $M = \norm{u}_{\infty}$. We have $\pb{\gd} \geq 1-\delta$. According to Lemma \ref{Lem:samplOne}, if $\gd$ holds, then $s$ is sampled from $\dis{u}$. In this case, the value $i$ computed at line 4 is distributed according to $\frac{\norm{g(x+e)_{I_s} - g(x)_{I_s}}_1}{\norm{u}_1} \cdot \frac{\abs{\gl(x+e)_i - \gl(x)_i}}{\norm{\gl(x+e)_{I_s} - \gl(x)_{I_s}}_1} = \frac{\abs{\gl(x+e)_i - \gl(x)_i}}{\norm{\gl(x+e) - \gl(x)}_1}$ (since $\norm{u}_1 = \norm{g(x+e) - g(x)}_1$). This corresponds to $\dis{g(x+e) - g(x)}$. Consequently, using Fact \ref{Fact:estBias}, $\norm{\ex{\wdi | x, e, \Gamma, \gd} - (g(x+e) - g(x))}_1 = \norm[\big]{\frac{\Gamma}{\norm{g(x+e) - g(x)}_1}(g(x+e) - g(x)) - (g(x+e) - g(x))}_1 = \abs{\Gamma - \norm{g(x+e) - g(x)}_1} \leq \eps/6 \norm{g(x+e) - g(x)}_1 \leq \eps$. Thus, $\norm{\ex{\wdi | x, e, \gd} - (g(x+e) - g(x))}_1 \leq \eps$. Moreover, $\norm{\wdi}_2 = \Gamma \leq (1+\eps/6) \norm{g(x+e) - g(x)}_1 \leq 7$. Finally, line 1 takes time $\so{\sqrt{k} \cdot \log(1/\delta) \cdot \eo}$, line 2 takes time $\so{\sqrt{k}/\eps \cdot \log(1/\delta) \cdot \eo}$, line 3 takes time $\bo{\sqrt{k} \cdot \eo}$, and lines 4 and 5 take time $\so{\eo}$.
\end{proof}

Finally, we analyze the cost of using the two above procedures in the subgradient descent algorithm studied in Section \ref{Sec:frame}. Unlike in the previous section, the time complexities $\costg$ and $\costd$ of $\qgs$ and $\qgds$ depend on the accuracy $\eps$. Consequently, it is more efficient to combine $\qgs$ with the classical procedure $\cgds$ when $n^{-1/2} \leq \eps \leq n^{-1/6}$, and to use Theorem \ref{Thm:clSub} when $\eps \leq n^{-1/2}$.

\begin{theorem}
  \label{Thm:quSub}
  There is a quantum algorithm that, given a submodular function $F : 2^V \ra [-1,1]$ and $\eps > 0$, computes a set $\bar{S}$ such that $\ex{{F}(\bar{S})} \leq \min_{S \subseteq V} F(S) + \eps$ in time $\so{n^{5/4}/\eps^{5/2} \cdot \log(\frac{1}{\eps}) \cdot \eo}$.
\end{theorem}

\begin{proof}
  We distinguish two cases, depending on the value of $\eps$. First, if $\eps \geq n^{-1/6}$, we instantiate Algorithm \ref{Algo:submDesc} with the quantum procedures $\qgs$ and $\qgds$ of Algorithms \ref{Algo:QGS} and \ref{Algo:QGDS} respectively. We choose the input parameters $T = \eps \sqrt{n}$, $N = 72^2n/\eps^2$, $\eps_0 = \eps/4$, $\eps_1 = \eps/8$ and $\delta = \eps/(8N)$. According to Corollary \ref{Cor:submDesc} and Corollary \ref{Cor:dataMaintain}, the run-time is $\so[\Big]{\frac{N}{T} \left(\sqrt{nT} + \frac{\sqrt{n}}{\eps} + \sum_{\tau = 1}^{T} \frac{\sqrt{\tau}}{\eps}\right) \cdot \allowbreak \log(1/\delta) \cdot \eo} = \so{n^{5/4}/\eps^{5/2} \cdot \log(1/\eps) \cdot \eo}$. Let us denote $\gd$ the event that \emph{all} calls to $\qgs$ and $\qgds$ in Algorithm \ref{Algo:submDesc} are correct (i.e. they satisfy Assumptions \ref{Assp:one} and \ref{Assp:two}). By the union bound, $\pb{\gd} \geq 1 - 2N\delta \geq 1 - \eps/4$. Moreover, according to Corollary \ref{Cor:submDesc}, the output $\bar{x}$ satisfies $\ex{f(\bar{x}) | \gd} \leq \fl(\xopt) + 3\eps/4$, where $\xopt \in \argmin_x f(x)$. Consequently, $\ex{f(\bar{x})} = \pb{\gd} \cdot \ex{f(\bar{x}) | \gd} + (1 - \pb{\gd}) \cdot \ex{f(\bar{x}) | \overline{\gd}} \leq 1 \cdot (\fl(\xopt) + 3\eps/4) + \eps/4 \cdot 1 \leq \fl(\xopt) + \eps$. Finally, using Proposition \ref{Prop:lv}, we can convert $\bar{x}$ into a set $\bar{S} \subseteq V$ such that $\ex{{F}(\bar{S})} \leq \min_{S \subseteq V} F(S) + \eps$ in time $\bo{n \log n + n \cdot \eo}$.

  If $\eps \leq n^{-1/6}$, we instantiate Algorithm \ref{Algo:submDesc} with the quantum procedure $\qgs$ of Algorithm \ref{Algo:QGS} and the classical procedure $\cgds$ of Algorithm \ref{Algo:CGDS}. We choose the input parameters $T = n^{1/4}/\eps^{1/2}$, $N = 54^2n/\eps^2$, $\eps_0 = \eps/3$, $\eps_1 = 0$ and $\delta = \eps/(3N)$. The run-time is $\so[\Big]{\frac{N}{T} \left(\left(\sqrt{nT} + \frac{\sqrt{n}}{\eps}\right)\log\frac{1}{\delta} + \sum_{\tau = 1}^{T} \tau\right) \cdot \eo} = \so{n^{5/4}/\eps^{5/2} \cdot \log(1/\eps) \cdot \eo}$. The proof of correctness is similar to the above paragraph.
\end{proof}

%% file: App_Counterexample.tex
Let $n=2$, let $F:2^{[2]}\rightarrow [-1,1]$ be a submodular function and let $f:[0,1]^2\rightarrow \R$ be its \lv\ extension. By definition (see Section \ref{Sec:Prelim}), the \lv\ subgradient of $f$ at $x = (x_1,x_2) \in [0,1]^2$ is
  \[
    g(x_1,x_2) =
      \begin{cases}
        \big(F(\{1\})-F(\varnothing),\,F(\{1,2\})-F(\{1\})\big) & \text{if } x_1\ge x_2\\
        \big(F(\{1,2\})-F(\{2\}),\,F(\{2\})-F(\varnothing)\big) & \text{if } x_1 < x_2.
      \end{cases}
  \]

Now let us consider a particular submodular function $F$. Let $F(\varnothing) = 0$, $F(\{1\})=-\frac{1}{2}$, $F(\{2\}) = 0$, $F(\{1,2\})=-1$, and therefore
  \[
    g(x_1,x_2) =
      \begin{cases}
        (-\frac{1}{2},\,-\frac{1}{2}) & \text{if } x_1\ge x_2\\
        (-1,\,0) & \text{if } x_1 < x_2.
      \end{cases}
  \]

The stochastic subgradient descent starts at $x\super{0}=(0,0)$, for which we have $g(x\super{0})=(-\frac{1}{2},-\frac{1}{2})$ and $\|g(x^{(0)})\|_1=1$. Hence, with probability $1/2$ each, we either have $\wg\super{0}=(-1,0)$ or $\wg\super{0}=(0,-1)$. Since $x\super{1} := \argmin_{x \in [0,1]^n} \norm{x - (x\super{0} - \eta \wg\super{0})}_2$ (for some step size parameter $0 < \eta < 1$), it follows that $x\super{1}=(\eta,0)$ or $x\super{1}=(0,\eta)$, respectively. Suppose the latter is the case: $\wg\super{0}=(0,-1)$ and $x\super{1}=(0,\eta)$. Then, the random variable $\wdi\super{1}$ is an estimate of the subgradient difference $d\super{1}=g(x\super{1})-g(x\super{0})=(-\frac{1}{2},\frac{1}{2})$, namely, $\ex{\wdi\super{1} | x\super{1}}=d\super{1}$. Now, observe that the random variable $\wg\super{1}=\wg\super{0}+\wdi\super{1}$ satisfies
  \[
  \ex{\wg\super{1}|x\super{1}} =
  \ex{\wg\super{0}|x\super{1}} + \ex{\wdi\super{1}|x\super{1}} =
  \wg\super{0} + d\super{1} = (-1/2,\,-1/2) \ne (-1,0) = g(x\super{1}).
  \]
Hence, the procedure in \cite{CLSW17c} that returns $\wg\super{t}$ is not a valid subgradient oracle. We note that $x\super{1}=(\eta,0)$ would have led to the same error. This problem can be fixed by defining $\wg\super{1}=\wwg\super{0}+\wdi\super{1}$ where $\wwg\super{0}$ is a second estimate of $g(x\super{0})$ \emph{independent} of $\wg\super{0}$, for instance $\wwg\super{0}=(-1,0)$ or $\wwg\super{0}=(0,-1)$ with probability $1/2$ each. In this case, we have $\ex{\wg\super{1} | x\super{1}} = g\super{0} + d\super{1} = g(x\super{1})$.